\documentclass[12pt, draftclsnofoot, onecolumn]{IEEEtran}

\usepackage{cite}
\usepackage{amsmath,amssymb,amsfonts,bm,color,amsthm}
\usepackage{algorithm,algorithmic}
\usepackage{graphicx}
\usepackage{setspace,multicol,multirow,subfigure,float}
\usepackage{enumitem}

\allowdisplaybreaks[4]

\newtheorem{pro}{Proposition}
\newtheorem{cor}{Corollary}

\begin{document}
	
\title{Massive Unsourced Random Access: Exploiting Angular Domain Sparsity}

\author{Xinyu~Xie,~
	Yongpeng~Wu,~\IEEEmembership{Senior~Member,~IEEE,}
	Jianping~An,~\IEEEmembership{Member,~IEEE,}
	Junyuan~Gao,~
	Wenjun~Zhang,~\IEEEmembership{Fellow,~IEEE,}
	Chengwen~Xing,~\IEEEmembership{Member,~IEEE,}
	Kai-Kit~Wong,~\IEEEmembership{Fellow,~IEEE,}
	and~Chengshan~Xiao,~\IEEEmembership{Fellow,~IEEE}
	\thanks{X. Xie, Y. Wu, J. Gao, and W. Zhang are with the Department of Electronic Engineering, Shanghai Jiao Tong University, Minhang 200240, China (e-mail: \{xinyuxie, yongpeng.wu, sunflower0515, zhangwenjun\}@sjtu.edu.cn).}
	\thanks{C. Xing and J. An are with the School of Information and Electronics, Beijing Institute of Technology, Beijing 100081, China (e-mail: xingchengwen@gmail.com, an@bit.edu.cn).}
	\thanks{K.-K. Wong is with the Department of Electronic and Electrical Engineering, University College London, London WC1E 6BT, U.K. (e-mail: kai-kit.wong@ucl.ac.uk).}
	\thanks{C. Xiao is with the Department of Electrical, and Computer Engineering, Lehigh University, Bethlehem, PA 18015 USA (e-mail: xiaoc@lehigh.edu).}
	\thanks{Corresponding authors: Y. Wu and J. An.}}

\maketitle

\begin{abstract}

This paper investigates the unsourced random access (URA) scheme to accommodate numerous machine-type users communicating to a base station equipped with multiple antennas. Existing works adopt a slotted transmission strategy to reduce system complexity; they operate under the framework of coupled compressed sensing (CCS) which concatenates an outer tree code to an inner compressed sensing code for slot-wise message stitching. We suggest that by exploiting the MIMO channel information in the angular domain, redundancies required by the tree encoder/decoder in CCS can be removed to improve spectral efficiency, thereby an uncoupled transmission protocol is devised. To perform activity detection and channel estimation, we propose an expectation-maximization-aided generalized approximate message passing algorithm with a Markov random field support structure, which captures the inherent clustered sparsity structure of the angular domain channel. Then, message reconstruction in the form of a clustering decoder is performed by recognizing slot-distributed channels of each active user based on similarity. We put forward the slot-balanced \mbox{$ K $-means} algorithm as the kernel of the clustering decoder, resolving constraints and collisions specific to the application scene. Extensive simulations reveal that the proposed scheme achieves a better error performance at high spectral efficiency compared to the CCS-based URA schemes.

\end{abstract}


\section{Introduction}

The next generation of cellular technology aims at wirelessly interconnecting sensors, machines, and wearable biomedical devices for potential new applications, thereby forming the architecture of the Internet of Things (IoT). Massive machine-type communications (mMTC) \cite{DSG17}, also known as massive access \cite{WGZ20,CNY21}, is a key requirement for IoT. Different from human-type communications (HTC), generic mMTC scenarios seek to establish reliable communications for a burgeoning number of devices with sporadic traffic patterns and small data payloads. Hence, it calls for novel theories and paradigms for the design of efficient multiple-access schemes.

Applying conventional grant-based random access (RA) schemes \cite{HHN13} to mMTC systems will reveal much energy consumption and high latency. Thus, grant-free RA protocols \cite{SL18} recently attract significant attention, where users directly send data to the base station (BS) without waiting for any approval. A typical type of grant-free RA scheme is based on the allocation of pilot sequences, where unique pilots as user identities are used for activity detection (AD) and channel estimation (CE) in the first stage \cite{LY18}. Data transmission is executed in the next stage using efficient RA schemes like sparse code multiple access (SCMA) \cite{LWS20}. As a prospective grant-free scheme for mMTC, a novel modality of \emph{unsourced random access} (URA) is introduced in \cite{P07}. Different from the pilot-based RA scheme, the URA users compulsorily utilize the same codebook to transmit messages directly without revealing their identities. Therefore, the BS only needs to acquire a list of transmitted messages without associating them to specific active users. Relying on the infinite block-length assumption, traditional asymptotic information theory provides limited perspectives to analyze the capacity of URA channels that propagate small user payloads. Therefore, in \cite{P07}, the author derives a finite block-length (FBL) achievability bound attained by random coding and maximum-likelihood (ML) decoding. Conventional RA schemes like ALOHA and treating inference as noise (TIN) are shown to remain an important gap to the FBL benchmark, thereby arousing great interest in introducing more efficient schemes.

An intuitive URA scheme can be easily obtained where a unique signature (codeword) is allotted to each message for transmission, and the BS performs AD to the set of codewords. Although it is evident that RA as an AD problem is closely related to the compressed sensing (CS) recovery problem \cite{CSY18,LY18}, directly applying CS techniques are prohibited in practical situations because the codebook size grows exponentially to the user payloads (e.g., to transmit $ 100 $ bits, one must assign $ 2^{100} $ signatures). Many practical URA coding schemes, e.g. \cite{CT20,RO21,FJC21}, have been studied on the additive white Gaussian noise (AWGN) channel to approach the FBL bound. They follow a recently proposed concatenated coding scheme termed coded compressed sensing (CCS) \cite{ACN20'1}, which couples an outer tree code and an inner CS code. More specifically, the entire message is partitioned into several smaller fragments, coupled by appending parity check bits generated from a linear block code. Each fragment is encoded by one column of a common coding matrix. The decoder first reconstructs transmitted fragments in all transmission slots, then relies on a tree-based decoding process to stitch these fragments together. An enhanced decoding strategy is reported in \cite{ACN20'2,ADP20}, where message stitching is executed right after the inner decoder recovers active fragments in each transmission slot. Existing fragment combinations impose restrictions on potential parity patterns, which helps narrow down the search realm for the CS algorithm in the next AD stage, leading to a systematic improvement in detection and decoding error probabilities. The works of \cite{FHJ19,FHJ21} further extend the CCS-based URA model to the Rayleigh block-fading AWGN channel in a MIMO setting, where a covariance-based support estimation method \cite{HJC18} is investigated for AD. Such a non-Bayesian method outperforms the approximate message passing (AMP) based Bayesian approach \cite{CSY18,LY18} in terms of AD error probability since it well exploits the channel hardening effect. However, due to redundancies required by the tree encoder/decoder for message stitching, the coding rate and spectral efficiency of CCS-based URA schemes are decreased. Other transmission schemes for MIMO URA can be found in \cite{FJC21'2,DLG20,DLG21}. A pilot-aided URA scheme is proposed in \cite{FJC21'2} based on pilot transmission with subsequent CE and maximum-ratio-combining (MRC). Such a protocol appears to be similar to the conventional two-stage design of pilot-based RA, while the difference is that pilot sequences in \cite{FJC21'2} are chosen pseudo-randomly from a common pilot pool based on the first few bits of active users' message.\footnote{A SCMA based URA scheme can be similarly designed, where the pilot for joint AD and CE in the first stage and the SCMA coding matrix for data transmission in the next stage are both chosen from a common pool based on the first few data bits. However, it is difficult to directly apply SCMA to the CCS scheme since it requires carefully designed pilot sequences to remove the scaling and permutation ambiguities in the blind detection process known as a dictionary learning problem.} Tensor-based modulation (TBM) is introduced to URA in \cite{DLG20,DLG21}, where data decoding is based on tensor decomposition and single-user demapping.

Aiming at decoupling the CCS structure, the authors in \cite{SBM21} suggest that the strong-correlation between slot-wise MIMO channels belonging to each active user enables the message recombination across slots. Specifically, after AD and CE, the determined active fragments are regrouped to the original packets by a clustering decoder capturing the similarity of their corresponding channels. Since the entire transmission frame is dedicated to data communication without redundancies, this uncoupled compressed sensing (UCS) scheme manifests high spectral efficiency. However, the correlation-aware clustering process counts on fractional parameters drawn from the well-estimated channels, while arguments like large-scale fading coefficients (LSFCs) are dropped. Also, lacking a collision resolution mechanism, one must apply a relatively large-sized codebook to reduce the probability of codeword collision (i.e., two or more users choose to send the same codeword at the same slot), which results in a huge computational burden. 

Massive MIMO technology, which utilizes a large number of antennas at the BS, provides high spatial resolution within the same time/frequency resource to support more active devices. To fully exploit rich spatial statistics reserved in the large-scale antenna space, we appeal to the \emph{angular domain channel} when modifying the UCS transmission scheme. The sparse nature of the angular domain channel \cite{GDW15, KGW20} promotes the sparsity of the CS paradigm, so less number of measurements are required to achieve the same level of estimation accuracy. Moreover, provided that angle of arrival (AoA) intervals of conflicting users are non-overlapping, codeword collision can be resolved \cite{YGX15}. We summarize the main contributions of the proposed uncoupled URA transmission scheme as follows.
\begin{itemize}[leftmargin=*]
	\item
	\textbf{A novel CS algorithm for AD and CE considering correlated angular domain channels:} We obey the generalized approximate message passing (GAMP) \cite{R11} framework for sparse signal reconstruction, where a Markov random field (MRF) \cite{SS11} structure is introduced to capture the inherent clustered sparsity of angular domain channels. We further provide an expectation-maximization (EM) way to learn crucial channel parameters dynamically. The proposed algorithm named \emph{EM-MRF-GAMP} achieves better CE accuracy compared to state-of-the-art CS techniques.
	\item
	\textbf{Clustering-based message recombination design tailored for angular domain channels:} We rely on unique angular transmission features reserved in the recovered channels to stitch the slot-distributed sequences together in a clustering way, thereby eliminating the tree-based encoding/decoding processes involved in CCS. The proposed \emph{slot-balanced $ K $-means} algorithm as the kernel of the clustering decoder enforces two constraints specific to the application scene. An adjustment is further made to alleviate the influence of codeword collision.
	\item
	\textbf{Uncoupled transmission design for URA with high spectral efficiency:} We leverage distinctive MIMO channel information rather than parity check bits to concatenate segmented data, which decouples the CCS scheme and achieves a higher coding rate. Compared to CCS-based URA regimes, the proposed uncoupled transmission scheme exhibits advantages with respect to decoding error probability in a high spectral efficiency region.
\end{itemize}

We organize the rest of this paper as follows. We describe the virtual angular domain channel model and the URA system model in the next section. In Section III, we overview the encoding and decoding processes of the UCS scheme exploiting angular domain sparsity. In Section IV, the EM-MRF-GAMP algorithm is put forward as the CS decoder. In Section V, we introduce the slot-balanced $ K $-means algorithm as the kernel of the clustering decoder. Numerical results of the system performance are presented in Section VI, followed by concluding remarks drawn in Section VII.

\emph{Notations:} Throughout this paper, the $ j $-th column and $ i $-th row of matrix $ \mathbf{X} $ are represented by $ \mathbf{x}_{j} $ and $ \mathbf{x}_{i,:} $, respectively, and the $ (i, j) $-th entry of $ \mathbf{X} $ is expressed by $ x_{i j} $. $ \mathbf{I}_{M} $ denotes the $ M $-dimensional identity matrix. We signify the conjugate, transpose, and conjugate transpose by superscripts $ (\cdot)^*,(\cdot)^T $, and $ (\cdot)^H $, respectively. Given any complex variable or matrix, $ \Re \{ \cdot \} $ and $ \Im \{ \cdot \} $ return its real and imaginary part, respectively. We denote the Euclid norm of vector $ \mathbf{x} $ by $ \left\| \mathbf{x} \right\|  $; $ | \cdot | $, $ \left\| \cdot \right\|_{2} $, and $ \left\| \cdot \right\|_{F} $ stand for the absolute value, the $ \ell_{2} $-norm, and the Frobenius norm, respectively. $ |\mathcal{X}| $ calculates the number of elements in set $ \mathcal{X} $, and $ \mathcal{X} \setminus \mathcal{Y} $ represents the set $ \{z: z \in \mathcal{X},z \notin \mathcal{Y} \} $. For an integer $ X > 0 $, we use the shorthand notation $ [X] $ to represent the set $ \{ 1, 2, \dots, X\} $. $ \mathcal{N}( x ; \widehat{x}, \mu^{x}) $ denotes the Gaussian distribution of a random variable $ x $ with mean $ \widehat{x} $ and variance $ \mu^{x} $, and $ \mathcal{CN}( x ; \widehat{x}, \mu^{x}) $ represents the case of the complex Gaussian distribution.

\section{System Model}

\subsection{Sparse 3D-MIMO Channel Modeling}	

Consider a single-cell network system where many single-antenna users communicate to a BS through the uplink synchronizing scheme. The BS is equipped with a uniform planar array (UPA) of $ M = M_{\mathrm{v}} \times M_{\mathrm{h}} $ antennas, arranging $ M_{\mathrm{v}} $ antennas in the vertical direction and $ M_{\mathrm{h}} $ antennas in the horizontal direction. The channel matrix $ \widetilde{\mathbf{H}}_k \in \mathbb{C}^{M_{\mathrm{v}} \times M_{\mathrm{h}}} $ of the $ k $-th user corresponding to the UPA can be modeled as the sum of $ L $ propagation paths, i.e.,
\begin{align}
	\widetilde{\mathbf{H}}_{k} = \sum_{l=1}^{L} g_{k, l} \mathbf{e}_{\mathrm{v}} \left( \Omega_{k, l}^{\mathrm{v}} \right) \mathbf{e}^{T}_{\mathrm{h}} \left( \Omega_{k, l}^{\mathrm{h}} \right)
\end{align}
where $ g_{k, l} $ is the path gain of the $ l $-th path between the BS and the $ k $-th user. Moreover, the vertical steering vector $ \mathbf{e}_{\mathrm{v}} $ and the horizontal steering vector $ \mathbf{e}_{\mathrm{h}} $ are in turn given by
\begin{align}
	\! \mathbf{e}_{\mathrm{v}} \left( \Omega_{k, l}^{\mathrm{v}} \right) &= \frac{1}{\sqrt{M_{\mathrm{v}}}} \left[ 1, e^{-j 2 \pi \Omega_{k, l}^{\mathrm{v}}}, \dots, e^{-j 2 \pi ( M_{\mathrm{v}} - 1 ) \Omega_{k, l}^{\mathrm{v}}} \right]^{T} \\
	\! \mathbf{e}_{\mathrm{h}} \left( \Omega_{k, l}^{\mathrm{h}} \right) &= \frac{1}{\sqrt{M_{\mathrm{h}}}} \left[ 1, e^{-j 2 \pi \Omega_{k, l}^{\mathrm{h}}}, \dots, e^{-j 2 \pi ( M_{\mathrm{h}} - 1 ) \Omega_{k, l}^{\mathrm{h}}}\right]^{T}
\end{align}
where $ \Omega_{k, l}^{\mathrm{v}} = \Delta \cos ( \phi_{k, l} ) $, $ \Omega_{k, l}^{\mathrm{h}} = \Delta \sin ( \phi_{k, l} ) \cos ( \varphi_{k, l} ) $, $ \phi_{k, l} \in [-\pi / 2, \pi / 2] $ and $ \varphi_{k, l} \in [-\pi / 2, \pi / 2] $ are the elevation AoA and the horizontal AoA, respectively, and $ \Delta $ stands for the ratio of the distance between two adjacent antenna elements to the carrier wavelength. We consider a typical half-wavelength spaced antenna array in this paper, i.e., $ \Delta = 1 / 2 $.

The channel $ \widetilde{\mathbf{H}}_{k} $ can be transformed to the angular domain by
\begin{align}
	\mathbf{H}_{k} =  \sum_{l=1}^{L} g_{k, l} \left[ \mathbf{U}_{\mathrm{v}}^{H} \mathbf{e}_{\mathrm{v}} \left( \Omega_{k, l}^{\mathrm{v}} \right) \right] \left[ \mathbf{U}_{\mathrm{h}}^{H} \mathbf{e}_{\mathrm{h}} \left( \Omega_{k, l}^{\mathrm{h}} \right) \right]^{T} = \mathbf{U}_{\mathrm{v}}^{H} \widetilde{\mathbf{H}}_{k} \mathbf{U}_{\mathrm{h}}^{\ast} \label{DFTtrans}
\end{align}
where 
\begin{align}
	\mathbf{U}_{\mathrm{v}} &= \left[ \mathbf{e}_{\mathrm{v}} ( 0 ), \mathbf{e}_{\mathrm{v}} \left( \tfrac{1}{M_{\mathrm{v}}} \right), \dots, \mathbf{e}_{\mathrm{v}} \left( \tfrac{M_{\mathrm{v}}-1}{M_{\mathrm{v}}} \right) \right] \\
	\mathbf{U}_{\mathrm{h}} &= \left[ \mathbf{e}_{\mathrm{h}} ( 0 ), \mathbf{e}_{\mathrm{h}} \left( \tfrac{1}{M_{\mathrm{h}}} \right), \dots, \mathbf{e}_{\mathrm{h}} \left( \tfrac{M_{\mathrm{h}}-1}{M_{\mathrm{h}}} \right) \right]
\end{align}
are discrete Fourier transform (DFT) matrices whose columns can be regarded as receive beamforming vectors that decompose the total transmit signal into multi-beams along fixed directions. Each entry of the angular domain channel $ \mathbf{H}_{k} $ counts the aggregated energy along the associated receive beam. For convenience, we write $ \widetilde{\mathbf{H}}_{k} $ and $ \mathbf{H}_{k} $ in the $ M $-dimensional vector form as
\begin{align}
	\widetilde{\mathbf{h}}_{k} = \sum_{l = 1}^{L} g_{k, l} \mathbf{e} \left( \Omega_{k, l}^{\mathrm{h}} \right) \otimes \mathbf{e} \left( \Omega_{k, l}^{\mathrm{v}} \right), \ \mathbf{h}_{k} = \mathbf{U}^{H} \widetilde{\mathbf{h}}_{k}
\end{align}
where $ \otimes $ denotes the Kronecker product and $ \mathbf{U} = \mathbf{U}_{\mathrm{h}} \otimes \mathbf{U}_{\mathrm{v}} $ is a unitary matrix.

The angular domain representation $ \mathbf{H}_{k} $ is actually sparse since: 1) the BS is surrounded with few scatterers in the propagation environment \cite{GDW15,KGW20}; 2) the $ ( m_{\mathrm{v}}, m_{\mathrm{h}} ) $-th entry of $ \mathbf{H}_{k} $ has a significant magnitude only if there is a scatterer with mean elevation/horizontal AoA satisfying \eqref{v1} and \eqref{v2} at the same time (see Appendix A for explanation). Against finite number of propagation paths, the sparsity of the angular domain channel is further promoted with the growing number of receiving antennas. Moreover, due to angular spread of the scatterer, the dominant elements of $ \mathbf{H}_{k} $ often appear in clusters in both vertical and horizontal dimensions. Such a two-dimensional clustered sparsity structure of $ \mathbf{H}_{k} $ is illustrated in \mbox{Fig. \ref{channel}}.

\begin{figure}[h]
	\centering
	\includegraphics[width=7cm]{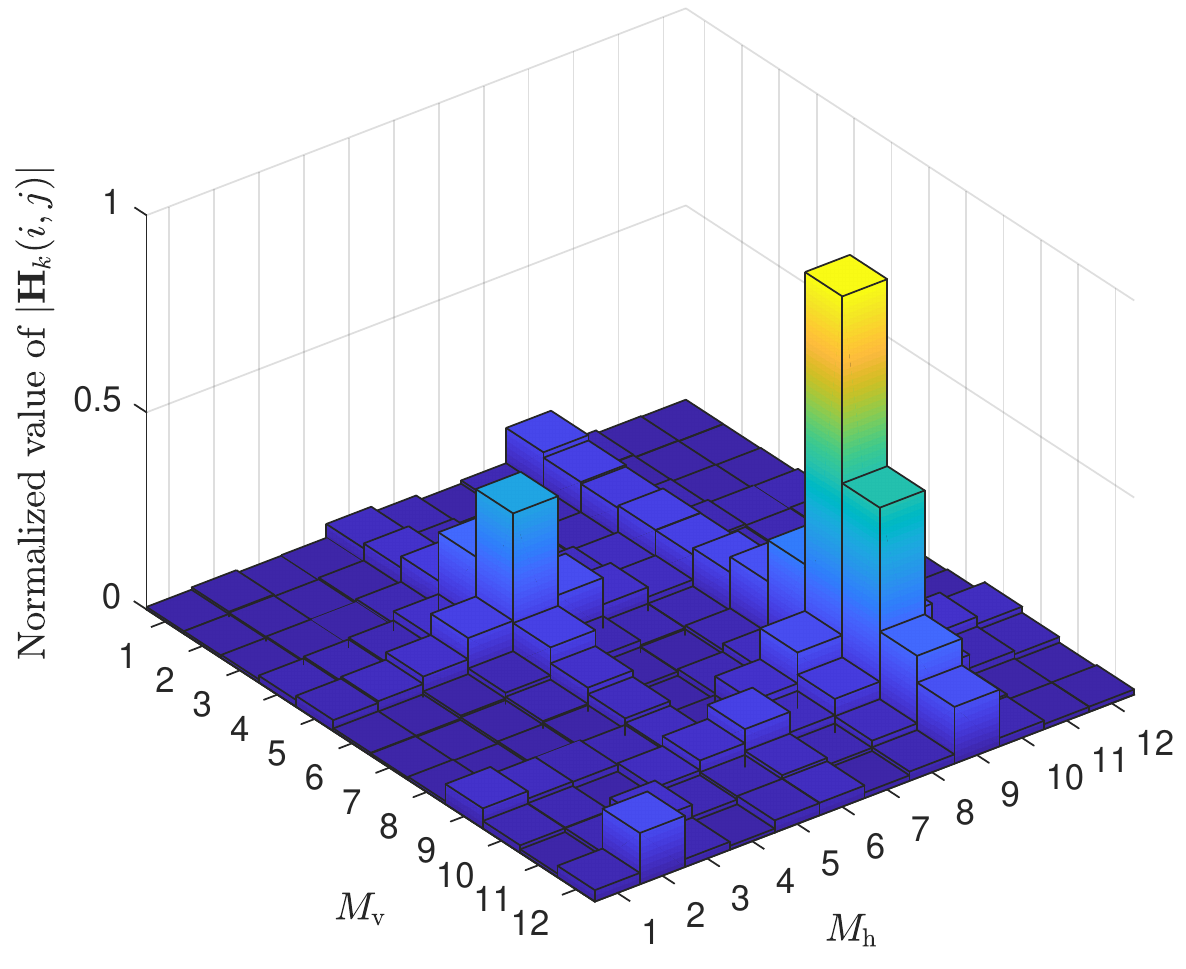}
	\caption{An example of the angular domain channel sparsity with an $ 12 \times 12 $ UPA. The channel is generated from a virtual 3D wireless channel model elaborated in Section VI. The maximum value of $ | \mathbf{H}_{k}( i,j ) | $ is normalized to 1.}
	\label{channel}
\end{figure}

\subsection{Signal Model}

The sporadic traffic pattern of mMTC indicates that only a small set of users $ \mathcal{K}_{\mathrm{a}} $ among a total number of $ K_{\mathrm{tot}} $ users are active. According to the URA setups, to communicate $ J $ bits of information to the BS, these active users pick up codewords $ \{ \widetilde{\mathbf{a}}_{i_{k}} \in \mathbb{C}^{N} : i_k \in [2^{J}], k \in \mathcal{K}_{\mathrm{a}} \} $ from a common codebook $ \widetilde{\mathbf{A}} = [ \widetilde{\mathbf{a}}_{1}, \dots, \widetilde{\mathbf{a}}_{2^{J}} ] \in \mathbb{C}^{N \times 2^{J}} $ to transmit. We set $ \widetilde{\mathbf{A}} \in \mathbb{C}^{N \times 2^{J}} $ a Gaussian independent and identically distributed (i.i.d.) matrix with each element $ a_{n j} \sim \mathcal{CN}( 0, 1 / N ) $, such that $ \mathbb{E} \{ \left\| \widetilde{\mathbf{a}} \right\|^{2} \} = 1 $. If we assume a block-fading channel where channel coefficients remain constant within the coherent block of $ N $ symbol transmissions, the received signal at each transmission slot takes on the form
\begin{align}
	\overline{\mathbf{Y}} = \sum_{k \in \mathcal{K}_{a}} \widetilde{\mathbf{a}}_{i_{k}} \widetilde{\mathbf{h}}_{k}^{T} + \overline{\mathbf{W}} =  \widetilde{\mathbf{A}} \mathbf{\Xi} \widetilde{\mathbf{H}} + \overline{\mathbf{W}}
\end{align}
where $ \mathbf{\Xi} \in \{ 0,1 \} ^{2^{J} \times K_{\mathrm{tot}}} $ is a codeword selection matrix with exactly one nonzero value at the $ i_k $-th entry of the $ k $-th column for $ k \in \mathcal{K}_{\text{a}} $, $ \widetilde{\mathbf{H}} = [ \widetilde{\mathbf{h}}_1, \dots, \widetilde{\mathbf{h}}_{K_{\mathrm{tot}}} ]^{T} \in \mathbb{C}^{K_{\mathrm{tot}} \times M} $, and $ \overline{\mathbf{W}} \in \mathbb{C}^{N \times M} $ is the matrix of additive white Gaussian noise with elements generated from an i.i.d. complex Gaussian distribution $ \mathcal{CN} \left( 0, 2 \sigma^{2} \right) $. The equivalent received signal in the angular domain can be expressed as
\begin{align}
	\widetilde{\mathbf{Y}} =  \widetilde{\mathbf{A}} \mathbf{\Xi} \widetilde{\mathbf{H}} \mathbf{U}^{*}  + \overline{\mathbf{W}} \mathbf{U}^{*} =  \widetilde{\mathbf{A}} \mathbf{\Xi} \mathbf{H} + \widetilde{\mathbf{W}} = \widetilde{\mathbf{A}} \widetilde{\mathbf{X}} + \widetilde{\mathbf{W}} \label{sigmod}
\end{align}
where $ \mathbf{H} = [\mathbf{h}_1, \dots, \mathbf{h}_{K_{\mathrm{tot}}}]^{T} $, $ \widetilde{\mathbf{W}} = \overline{\mathbf{W}} \mathbf{U}^{*} $ is the equivalent noise sample matrix, and $ \widetilde{\mathbf{X}} \triangleq \mathbf{\Xi} \mathbf{H} \in \mathbb{C}^{2^{J} \times M} $.

\section{Slotted Transmission Scheme for Unsourced Random Access}

Transmission protocol design for URA faces the bottleneck that if one wishes to send the entire message of length $ B $ (on the order of $ 100 $) within a single transmission slot, decoding will entail finding the support of $ 2^{B} $ possible codewords, which is computationally intractable. The recent introduction of CCS \cite{ACN20'1}, demonstrated in \mbox{Fig. \ref{treecode}}, takes a divide-and-conquer strategy to alleviate the system complexity. It utilizes a concatenated coding scheme coupling an outer tree code and an inner CS code. Each user payload of size $ B $ is transmitted using $ S' $ fragments of amenable length $ J $. Within each fragment, parity check bits are added after partitioned information bits in the form of an outer tree code; they are generated by pseudo-random linear combinations of information bits from previous fragments. Then, it is the task of the CS encoder to map each fragment (denoted by $ \mathbf{v} \in \{0,1\}^{J} $) to a codeword in the common codebook to emit over the noisy channel. The encoding process can be portrayed as the product of a common coding matrix $ \widetilde{\mathbf{A}} \in \mathbb{C}^{N \times 2^{J}} $ and an index vector $ \bm{\xi} \in \{0,1\}^{2^{J}} $. Such a vector associated with $ \mathbf{v} $ contains all zeros except one non-zero element at location $ \operatorname{decimal}(\mathbf{v}) $, where $ \operatorname{decimal}(\mathbf{v}) $ represents the radix ten equivalent of the binary vector $ \mathbf{v} $. In other words, fragment $ \mathbf{v} $ chooses the $ \operatorname{decimal}(\mathbf{v}) $-th column of $ \mathbf{A} $ as the codeword for transmission. After the BS determines the active codewords through a CS support recovery method, the tree-based outer decoder reconstructs the entire message by recombining the slot-wise fragments fitting exactly the parity check rules.

\begin{figure*}[!t]
	\centering
	\subfigure[]{\includegraphics[width=7.4cm]{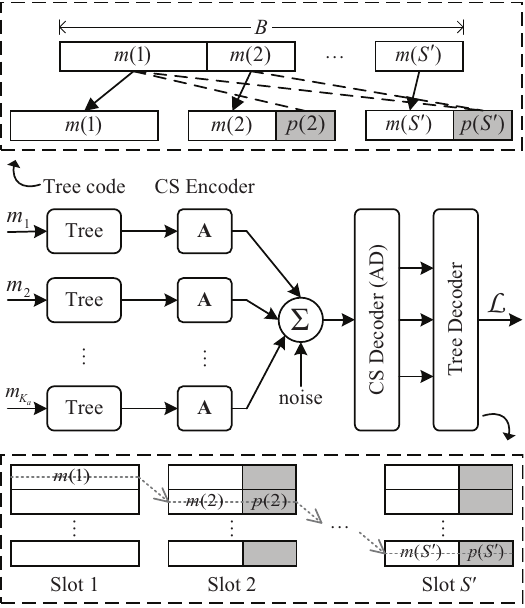}
		\label{treecode}}
	\hspace{0.5cm}
	\subfigure[]{\includegraphics[width=7.4cm]{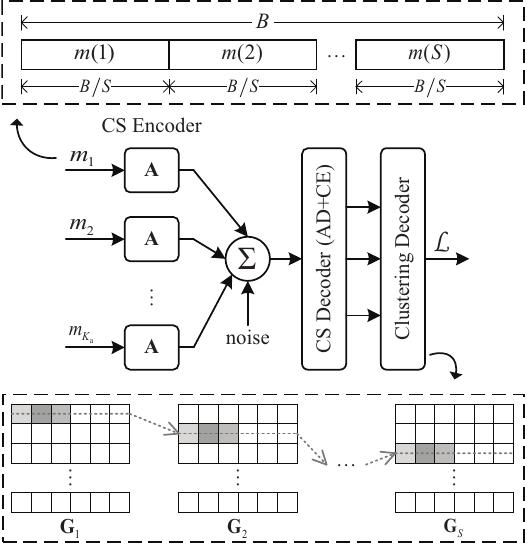}
		\label{cluscode}}
	\caption{Schematic diagrams of two transmission schemes for URA with $ m_{k} $ the transmitted information and $ \mathcal{L} $ the output message list: (a) the overall architecture of the CCS scheme, where the transmission frame is composed of portioned information bits $ m( s ) $ and added parity check bits $ p( s ) $; (b) the overall architecture of the proposed UCS scheme.}
\end{figure*}

We suggest that the unique angular propagation pattern indicated by angular domain channels pertaining to each active user already offers adequate information to regroup messages scattered among different transmission slots. User data propagates through different scatterers with different arriving angles and energies to the BS. Therefore, the sparsity and magnitude of each entry of the angular domain channel vector vary between users. These channel statistics are assumed to be almost unchanged within the short period of time when grant-free URA happens. Hence, the message stitching process can be rendered into distinguishing recovered channels of each active user from different slots.

The proposed uncoupled transmission scheme exploiting angular domain sparsity is illustrated in \mbox{Fig. \ref{cluscode}}. Without appending redundancies, the $ B $-bit message is divided into $ S = \lceil B / J \rceil $ fragments of length $ J $, each encoded by the CS encoder as discussed in the above context. After transmitting codewords over noisy channels, it is the task for the BS to determine active codewords and also retrieve their corresponding channels. Finally, slot-distributed codewords of each entire packet are recognized by a clustering decoder based on their similarity. In summary, the UCS regime forms three significant departures from CCS:
\begin{enumerate}[leftmargin=*]
	\item
	The data structure of CCS includes both information bits and parity check bits, while that of UCS contains only information bits without redundancies for concatenation.
	\item
	The inner decoder of CCS only retrieves the codeword activity pattern, whereas the CS decoder under the uncoupled framework also performs CE for these active codewords.
	\item
	The outer tree decoder in CCS is replaced in UCS by a clustering-based decoder.
\end{enumerate}
Considering that no redundancies are required to couple information bits across slots, the UCS scheme is foreseeable to manifest high spectral efficiency. In the next section, we elaborate the EM-MRF-GAMP algorithm operating as the CS decoder for AD and CE. And in Section V, the slot-balanced $ K $-means algorithm is addressed to enforce the clustering-based message stitching.

\section{Proposed Compressed Sensing Algorithm for Joint Activity Detection and Channel Estimation}

In this section, we present the EM-MRF-GAMP algorithm as the kernel of the CS decoder in UCS. First, we recognize the AD and CE problems in the CS recovery paradigm. Afterwards, under the Bayesian inference framework, the MRF model is introduced to model the underlying clustered support structure of the sparse angular domain channels. On such bases, we resort to the message passing strategy for sparse signal recovery. Finally, we leverage the EM framework to infer important channel parameters to help with the reconstruction.

\subsection{Activity Detection and Channel Estimation as a Compressed Sensing Problem}

Within this paper, the associated channel to the $ i $-th codeword is defined by the $ i $-th column of $ \widetilde{\mathbf{X}} $ expressed as
\begin{align}
	\widetilde{\mathbf{x}}_{i,:} = \sum_{k \in \mathcal{K}_{a}} \xi_{i, k} \mathbf{h}_{k}^{T} \label{xcolumn}
\end{align}
where $ \xi_{i, k} $ is the $ (i, k) $-th element of $ \mathbf{\Xi} $ in \eqref{sigmod}: it takes nonzero value only if $ i = i_{k} $, i.e., at least one user arranges to send the $ i $-th codeword, in which case, such a codeword is said to be ``active''. It is easily seen that to perform AD and CE to the set of codewords is to recover $ \widetilde{\mathbf{X}} $ from the noisy observation $ \widetilde{\mathbf{Y}} $ in \eqref{sigmod}. Since all users choose codewords independently and uniformly from the common codebook, $ \widetilde{\mathbf{x}}_{i,:} $ is identically zero with probability $ (1 - 2^{-J})^{K_{\mathrm{a}}} $. Given $ K_{\mathrm{a}} = \left| \mathcal{K}_{\mathrm{a}} \right|  \ll 2^{J} $, the matrix $ \widetilde{\mathbf{X}} $ is row-sparse. Profited from the sparse nature of the angular domain channel $ \mathbf{h}_{k} $, the sparsity of $ \widetilde{\mathbf{X}} $ is further encouraged within each row. Therefore, we see that the joint AD and CE problem is equivalent to a CS recovery problem.

Throughout this paper, we assign a Laplacian prior to each channel coefficient, as particulars will be discussed later. Since the Laplacian distribution is defined only over the real number field, we tune the complex-valued model \eqref{sigmod} to the following equivalent real-valued model:
\begin{align}
	\underbrace{\begin{bmatrix}
			\Re \{ \widetilde{\mathbf{Y}} \} \\[0.4em]
			\Im \{ \widetilde{\mathbf{Y}} \}
	\end{bmatrix}}_{\substack {\\ {\triangleq} \\ \mathbf{Y}}} =
	\underbrace{\begin{bmatrix}
			\Re \{ \widetilde{\mathbf{A}} \} & \!\!\!\! -\Im \{ \widetilde{\mathbf{A}} \} \\[0.4em]
			\Im \{ \widetilde{\mathbf{A}} \} & \!\!\!\! \Re \{ \widetilde{\mathbf{A}} \}
	\end{bmatrix}}_{\substack {\\ {\triangleq } \\ {\mathbf{A}}}}
	\underbrace {\begin{bmatrix}
			\Re \{ \widetilde{\mathbf{X}} \} \\[0.4em]
			\Im \{ \widetilde{\mathbf{X}} \}
	\end{bmatrix} }_{\substack {\\ {\triangleq} \\ {\mathbf{X}}}} +
	\underbrace {\begin{bmatrix}
			\Re \{ \widetilde{\mathbf{W}} \} \\[0.4em]
			\Im \{ \widetilde{\mathbf{W}} \}
	\end{bmatrix} }_{\substack {\\ {\triangleq } \\ {\mathbf{W}}}}.
\end{align}
For convenience, we divide the row index of $ \mathbf{X} $ (i.e., $ j \in [ 2^{J + 1} ] $) into two sets of sequences: the row index of the equivalent real part is denoted by $ j_{\mathrm{re}} \in [ 2^{J} ] $ and that of the imaginary part by $ j_{\mathrm{im}} = j_{\mathrm{re}} + 2^{J} $. Except the row sparsity inherited from $ \widetilde{\mathbf{X}} $, the matrix $ \mathbf{X} $ also possesses a group sparsity structure since $ \mathbf{x}_{j_{\mathrm{re}, :}} = \Re \{ \widetilde{\mathbf{x}}_{i, :} \} $ and $ \mathbf{x}_{j_{\mathrm{im}, :}} = \Im \{ \widetilde{\mathbf{x}}_{i, :} \} $ share the same active state.

\subsection{Probability Model}

We follow the Bayesian approach to retrieve $ \mathbf{X} $ from the received noisy superposition. For convenience, we represent the probability distribution function (pdf) of a true but unknown distribution by $ p_{0}(\cdot) $, and the postulated prior used for inference algorithm design by $ p(\cdot) $. First, we assign a zero-mean Laplacian distribution to each angular domain channel coefficient, i.e.,
\begin{align}
	\! p( h_{\mathrm{re}} ) = \frac{\lambda}{2} \exp \left( - \lambda | h_{\mathrm{re}} | \right), \ p( h_{\mathrm{im}} ) = \frac{\lambda}{2} \exp \left( - \lambda | h_{\mathrm{im}} | \right)
\end{align}
where $ h_{\mathrm{re}} $ and $ h_{\mathrm{im}} $ are the real and imaginary part of the channel coefficient $ h $, respectively, and $ \lambda $ is a scale parameter known as the Laplace rate. The motivation comes from \cite{BSY19} where the authors suggest employing Laplacian distributed random variables to model the MIMO mmWave channel coefficients in the angular domain, which are obtained by a DFT transformation \cite{S02} similar to our case in \eqref{DFTtrans}. It is found in \cite{BSY19} that the designed Bayes-optimal channel estimator under a Laplacian prior exhibits improvements in channel estimation accuracy and convergence rate compared to the Gaussian mixture prior \cite{VS13}. Subsequently, we give a Bernoulli-Laplacian prior distribution to each entry of the sparse matrix $ \mathbf{X} $, represented as
\begin{align}
	p( x_{j m} | b_{j' m} ) =  \dfrac{\lambda}{2} \exp \left( - \lambda \left| x_{j m} \right| \right) \delta(b_{j' m} - 1) + \delta(x_{j m}) \delta(b_{j' m} + 1)
\end{align}
where $ \delta(\cdot) $ denotes the Dirac function, and $ b_{j' m} \in \{-1, 1\} $ with $ j' = j - 2^{J} \lfloor j / 2^{J} \rfloor $ is a binary state capturing the group support structure of the real and imaginary part of the complex $ \widetilde{x}_{j' m} $; $ b_{j' m} = \pm 1 $ signifies that both $ x_{j' m} = \Re \{\widetilde{x}_{j' m}\} $ and $ x_{j' + 2^{J}, m} = \Im \{\widetilde{x}_{j' m}\} $ are nonzero/zero.

We take into account the clustered support structure of the angular domain channel coefficients by leveraging an MRF prior at the active state side. The motivation comes from the widely application of the MRF prior in modeling two-dimensional block-sparse image signals in many image recovery methods \cite{ZYH20}. Such a prior has the potential to encourage clustered sparsity and suppress ``isolated coefficients'' whose activity pattern is different from that of other coefficients. To model the hidden binary state of the channel of the $ j' $-th codeword, denoted by $ \mathbf{b}_{j', :} = [b_{j' 1}, \dots, b_{j' M}] \in \{-1, 1\}^{1 \times M} $, we employ an Ising model \cite{SS11}, i.e., $ p( \mathbf{b}_{j', :} ) $ is obtained by sampling
\begin{align}
	\exp \left( \sum_{m = 1}^{M} \left( \frac{1}{2} \sum_{k \in \mathcal{R}_{m}} \beta_{j'} b_{j' k} - \alpha_{j'} \right) b_{j' m} \right) = \left( \prod_{m = 1}^{M} \prod_{k \in \mathcal{R}_{m}} \exp \left( \beta_{j'} b_{j' m} b_{j' k} \right) \right)^{\frac{1}{2}} \prod_{m = 1}^{M} \exp \left( -\alpha_{j'} b_{j' m} \right) \label{Ising}
\end{align}
at $ \mathbf{b}_{j', :} $, where $ \mathcal{R}_{m} \subset \{1, \dots, M\} \setminus m $ is the set of related entries of index $ m $. The Ising prior depicts the sparsity and the interaction between parameters of $ \mathbf{b}_{j'} $ by arguments $ \alpha_{j'} $ and $ \beta_{j'} $, respectively. A higher magnitude of $ \alpha_{j'} $ indicates a sparser activity pattern, and a larger value of $ \beta_{j'} $ heightens the covariance between related entries.

\begin{figure*}[!t]
	\centering
	\subfigure[]{\includegraphics[height=7.5cm]{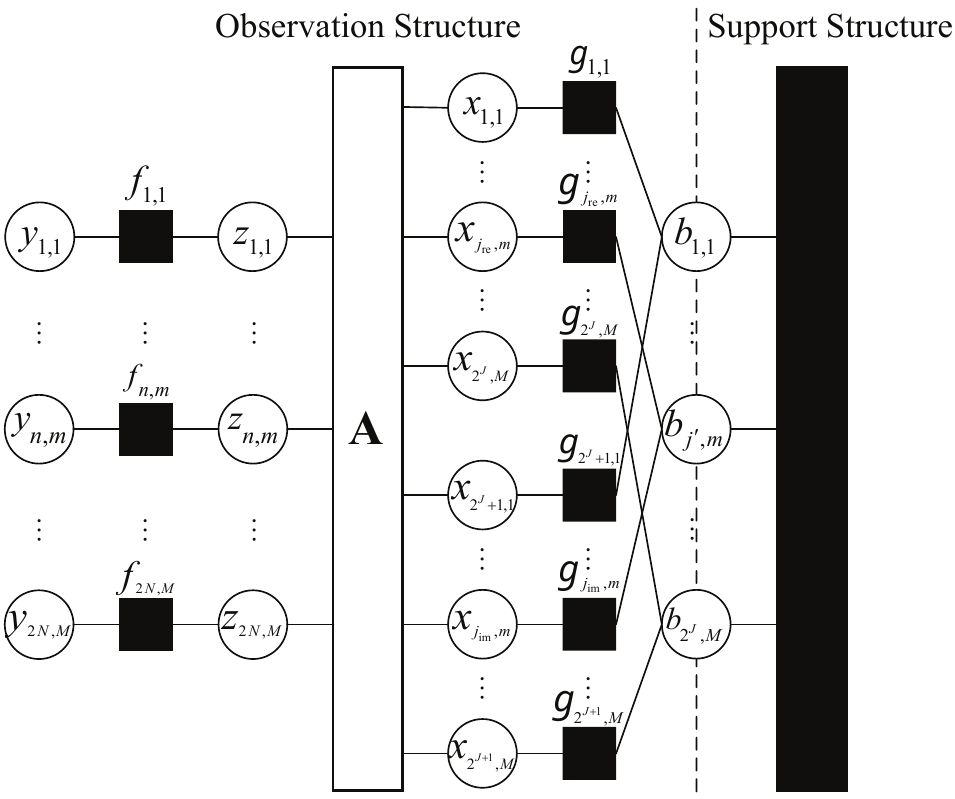}
		\label{GAMPfactor}}
	\subfigure[]{\includegraphics[height=7.5cm]{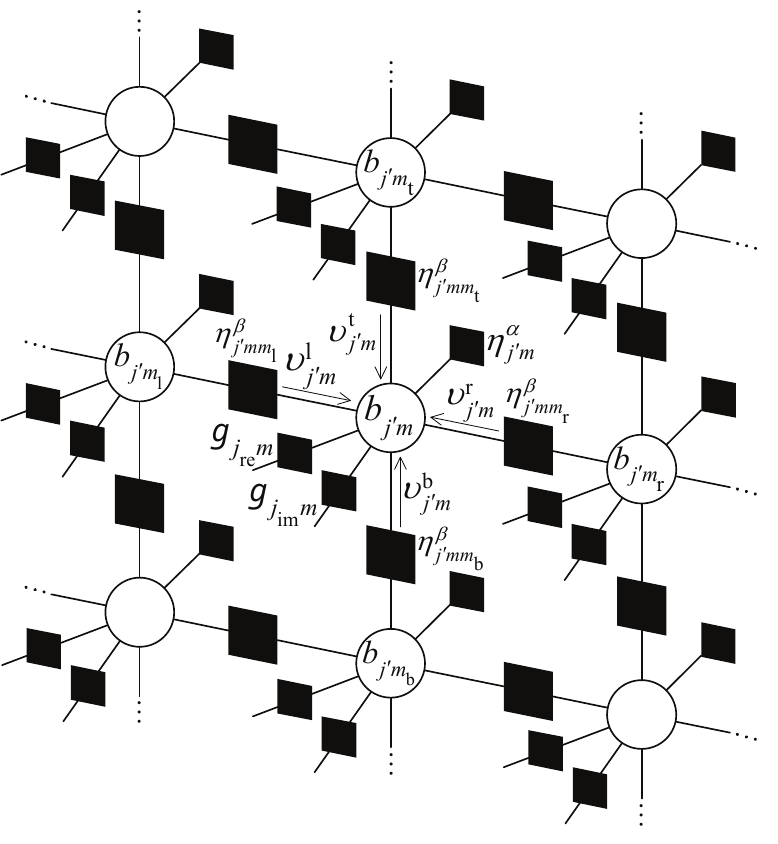}
		\label{MRFfactor}}
	\caption{Factor graphs associated to the model in \eqref{map}: (a) Factor graph for the hierarchical probability model in \eqref{map}, where the box marked `A' represents the process of RLM and adsorbs factor nodes $ \left\lbrace p( z_{n m} | \mathbf{a}_{n, :} \mathbf{x}_{m} ): n \in [N], m \in [M] \right\rbrace $; (b) Factor graph for the MRF support structure.}
	\label{factorgraph}
\end{figure*}

\begin{table}[!t]
	\renewcommand{\arraystretch}{1.5}
	\caption{Notations of Factor Nodes in \mbox{Fig. \ref{factorgraph}}}
	\label{notation}
	\centering
	\begin{tabular}{|c|c|c|}
		\hline
		\textbf{Factor} & \textbf{Distribution} & \textbf{Functional Form} \\
		\hline
		$ f_{l m} $ & $ p ( y_{l m} | z_{l m} ) $ & $ \frac{1}{\sqrt{2 \pi \sigma^{2}}} \exp \left( - \frac{(y_{l m} - z_{l m})^{2}}{2 \sigma^{2}} \right) $ \\
		\hline
		$ g_{j m} $ & $ p\left( x_{j m} | b_{j' m} \right) $ & $ \frac{\lambda}{2} \exp ( - \lambda | x_{j m} | ) \delta(b_{j' m} - 1) $ $ + \delta(x_{j m}) \delta(b_{j' m} + 1) $ \\
		\hline
		$ \eta_{j' m}^{\alpha} $ & $ \backslash $ & $ \exp \left( -\alpha_{j'} b_{j' m} \right) $ \\
		\hline
		$ \eta_{j' m k}^{\beta} $ & $ \backslash $ & $ \exp \left( \beta_{j'} b_{j' m} b_{j' k} \right) $ \\
		\hline
	\end{tabular}
\end{table}

Denote by $ \mathbf{B} = [ \mathbf{b}_{1, :}^{T}, \dots, \mathbf{b}_{2^{J}, :}^{T} ]^{T} \in \{-1, 1\}^{2^{J} \times M} $ the binary state matrix. To infer $ \mathbf{B} $ and $ \mathbf{X} $ from the observed signal $ \mathbf{Y} $, we derive the posterior probability density of $ \mathbf{B} $ and $ \mathbf{X} $ given $ \mathbf{Y} $ as
\begin{align}
	p \left( \mathbf{B}, \mathbf{X} | \mathbf{Y} \right) &\propto p\left( \mathbf{Y} | \mathbf{B}, \mathbf{X} \right) p\left( \mathbf{X} | \mathbf{B} \right) p\left( \mathbf{B} \right) \notag \\  &\propto \exp \left( -\frac{1}{\sigma^{2}} \|\mathbf{Y} - \mathbf{Z}\|_{2}^{2} \right) \prod_{j} p\left( \mathbf{x}_{j, :} | \mathbf{b}_{j', :} \right) p\left( \mathbf{b}_{j', :} \right) \label{map}
\end{align}
where $ \mathbf{Z} = \mathbf{A} \mathbf{X} $ is the output of a random linear mixing (RLM) transform \cite{R11} with $ \mathbf{X} $ the input. We demonstrate the connections of random variables in \eqref{map} by a factor graph as shown in \mbox{Fig. \ref{factorgraph}}, where circles and squares signify variable nodes and factor nodes, respectively. The notations of factor nodes are detailed in Table \ref{notation}. \mbox{Fig. \ref{GAMPfactor}} generally illustrates the  hierarchical probability model, and \mbox{Fig. \ref{MRFfactor}} concretely describes the MRF support estimation module, where we arrange the factor graph in a two-dimension shape corresponding to the UPA arrangement. Except for the nodes at edges, each support variable node $ b_{j' m} $ is linked to four adjacent nodes termed $ b_{j' m_{\mathrm{l}}} = b_{j', m - M_{\mathrm{v}}} $, $ b_{j' m_{\mathrm{r}}} = b_{j', m + M_{\mathrm{v}}} $, $ b_{j' m_{\mathrm{t}}} = b_{j', m - 1} $, and $ b_{j' m_{\mathrm{b}}} = b_{j', m + 1} $ from the left, right, top, and bottom, respectively.

Unfortunately, the minimum mean square error (MMSE) estimation or the maximum \emph{a posteriori} (MAP) estimation with respect to $ p ( \mathbf{B}, \mathbf{X} | \mathbf{Y} ) $ in \eqref{map} is hard to carry out straightforwardly. Especially when RLM occurs, it is computationally intractable to reach a precise posterior distribution form of any individual component $ x $ since it involves marginalizing a joint distribution with high dimensions. In order to obtain a tractable proxy, we refer to the framework of GAMP \cite{R11} and propose a message passing based method, as detailed in what follows.

\subsection{Message Passing Algorithm for Signal Reconstruction}

As a practical approach to tackle the RLM estimation problem, GAMP employs loopy belief propagation (BP) over the factor graph to make approximate inference of the marginal. For the MMSE estimation of $ p ( \mathbf{B}, \mathbf{X} | \mathbf{Y} ) $, the sum-product algorithm \cite{KFL01} is applied to reduce the number of messages involved in propagation. Furthermore, for a random Gaussian i.i.d. transformation matrix $ \mathbf{A} $ under the large system limit hypothesis, i.e., $ 2^{J} \to \infty $, the messages passed between the edges of the factor graph admit very good Gaussian approximations. This helps to simplify the loopy message passing process to iteratively update means and variances of Gaussian distributions, following the algorithmic description detailed in Algorithm \ref{GAMP}.

\begin{algorithm}[!t]
	\caption{EM-MRF-GAMP with Laplacian Prior}
	\label{GAMP}
	\begin{algorithmic}[1]
		\STATE \textbf{Input:} Observed signal $ \mathbf{Y} $, measurement matrix $ \mathbf{A} $, precision
		tolerance $ \tau $, maximum number of iterations $ T_{\mathrm{max}} $ and $ T_{\mathrm{mrf}} $
		\STATE \textbf{Initialize:} \\
		\begin{spacing}{1.2}
			$ \quad \forall n, m: \widehat{s}_{n m}( 0 ) = 0  $, $ \forall j, m: $ choose $ \widehat{x}_{j m}( 1 ) $, $ \mu_{j m}^{x}( t ) $, \\
			$ \quad \forall j': \alpha_{j'} = \beta_{j'} = 0.4 $, $ \forall j', m, d: \kappa_{j' m_{d}} = 0.5 $, $ \lambda = 1 $, $ \sigma^{2} = \frac{\| \mathbf{Y} \|_{F}^{2}}{2 M N ( R + 1 )} $
		\end{spacing} \vspace{0pt}
		\FOR{$ t = 1,2,\dots,T_{\mathrm{max}} $}
		\begin{spacing}{1.5}
			\STATE $ \forall n, m: \mu_{n m}^{p}( t ) = \sum_{j} \left| a_{n j} \right|^{2} \mu_{j m}^{x}( t ) $
			\STATE $ \forall n, m: \widehat{p}_{n m}( t ) = \sum_{j} a_{n j} \widehat{x}_{j m}( t ) - \mu_{n m}^{p}( t ) \widehat{s}_{n m}( t - 1 ) $
			\STATE $ \forall n, m: \mu_{n m}^{z}( t ) = \mu_{n m}^{p} \sigma^{2} / \left( \mu_{n m}^{p} + \sigma^{2} \right) $
			\STATE $ \forall n, m: \widehat{z}_{n m}( t ) = \left( \mu_{n m}^{p} y_{n m} + \sigma^{2} \widehat{p}_{n m} \right) / \left( \mu_{n m}^{p} + \sigma^{2} \right) $
			\STATE $ \forall n, m: \mu_{n m}^{s}( t ) = \left[ \mu_{n m}^{p}( t ) - \mu_{n m}^{z}( t ) \right] / \left[ \mu_{n m}^{p}( t ) \right]^{2} $
			\STATE $ \forall n, m: \widehat{s}_{n m}( t ) = \left[ \widehat{z}_{n m}( t ) - \widehat{p}_{n m}( t ) \right] / \mu_{n m}^{p}( t ) $
			\STATE $ \forall j, m: \mu_{j m}^{r}( t ) = \left[ \sum_{n} \left| a_{n j} \right|^{2} \mu_{n m}^{s}( t ) \right]^{-1} $
			\STATE $ \forall j, m: \widehat{r}_{j m}( t ) = \widehat{x}_{j m}( t ) + \mu_{j m}^{r}( t ) \sum_{n} a_{n j} \widehat{s}_{n m}( t ) $
		\end{spacing} \vspace{0pt}
		\STATE \textbf{\% MRF Support Estimation Module}
		\STATE $ \forall j,m $: Compute input $ \varpi_{j m}( t ) $ via \eqref{pi}
		\FOR{$ t_{\mathrm{mrf}} = 1,2,\dots,T_{\mathrm{mrf}} $}
		\STATE $ \forall j,m $: Update $ \nu_{jm}^{\mathrm{l}} $, $ \nu_{jm}^{\mathrm{r}} $, $ \nu_{jm}^{\mathrm{t}} $ and $ \nu_{jm}^{\mathrm{b}} $ via \eqref{MRFmp}
		\ENDFOR
		\STATE $ \forall j,m $: Compute output $ \rho_{j m} ( t ) $ via \eqref{rho}
		\begin{spacing}{1.3}
			\STATE $ \forall j, m: \widehat{x}_{j m}( t + 1 ) = \mathbb{E} \left\{ x_{j m} | \mathbf{Y} ; \widehat{r}_{j m}( t ), \mu_{j m}^{r}( t ), \rho_{j m}( t ), \lambda \right\} $
			\STATE $ \forall j, m: \mu_{j m}^{x}( t + 1 ) = \operatorname{Var} \left\{ x_{j m} | \mathbf{Y} ; \widehat{r}_{j m}( t ), \mu_{j m}^{r}( t ), \rho_{j m}( t ), \lambda \right\} $
		\end{spacing} \vspace{0pt}
		\STATE \textbf{\% EM Update}
		\STATE Update $ \sigma^{2} $ and $ \lambda $ via \eqref{sigma} and \eqref{lambda}, respectively
		\STATE \textbf{if} $ \| \widehat{\mathbf{X}}( t + 1 ) - \widehat{\mathbf{X}}( t ) \|_{F}^{2} < \tau \| \widehat{\mathbf{X}}( t ) \|_{F}^{2} $, \textbf{stop}
		\ENDFOR
		\STATE \textbf{Output:} Estimated signal $ \widehat{\mathbf{X}} $
	\end{algorithmic}
\end{algorithm}

The message to variable node $ z_{n m} $ from the RLM output side is calculated by integrating $ p( z_{n m} | \mathbf{a}_{n, :} \mathbf{x}_{m} ) $ over all the variable nodes related to the elements in $ \mathbf{x}_{m} $. According to the central limit theorem, such a calculation has a Gaussian approximation $ \mathcal{N}( z_{n m} ; \widehat{p}_{n m}, \mu_{n m}^p ) $ with variance $ \mu_{n m}^p $ and mean $ \widehat{p}_{n m} $ obtained from lines 4 and 5 in Algorithm \ref{GAMP}, separately. Note that an equivalent ``Onsager'' correction term \cite{DMM09} is appended when computing the Gaussian mean. For an AWGN channel, the mean and variance of the marginal posterior $ p_{0}(z_{n m} | \mathbf{Y}) $ can be approximated by an empirical calculation involving the product of two Gaussian distributions (see lines 6 and 7 of Algorithm \ref{GAMP}). Then, the scaled residual $ \widehat{s}_{n m} $ and the inverse-residual-variance $ \mu_{n m}^{s} $ are computed as detailed in lines 8 and 9. Finally, the inverse output message to variable node $ x_{j m} $ is also approximately Gaussian with mean $ \widehat{r}_{j m} $ and variance $ \mu_{n m}^{r} $ (see lines 10 and 11 of Algorithm \ref{GAMP}).

Now we concentrate on the message passing concerning the MRF support estimation module. In GAMP, the message from variable node $ x_{j m} $ to factor node $ f_{j m} $ takes on the same form as the RLM inverse output message, i.e., $ \nu_{x_{j m} \to g_{j m}} = \mathcal{N} ( x_{j m} ; \widehat{r}_{j m}, \mu^{r}_{j m} ) $.
In Appendix B, we derive the message from $ g_{j m} $ to $ b_{j' m} $ as
\begin{align}
	\nu_{g_{j m} \to b_{j' m}} = \varpi_{j m} \delta( b_{j' m} - 1 ) + ( 1 - \varpi_{j m} ) \delta( b_{j' m} + 1 )
\end{align}
with
\begin{align}
	\varpi_{j m} = \dfrac{I_{x}^{-} + I_{x}^{+}}{\mathcal{N} ( 0 ; \widehat{r}, \mu^{r} ) + ( I_{x}^{-} + I_{x}^{+} )} \label{pi}
\end{align}
where
\begin{align}
	I_{x}^{-} &= \dfrac{\lambda}{2} \exp \left( \dfrac{1}{2} \lambda^2 \mu^{r}_{j m} + \lambda \widehat{r}_{j m} \right) \Phi_{\mathcal{N}} \left( \tfrac{-\widehat{r}^{-}_{j m}}{\sqrt{\mu^{r}_{j m}}} \right) \\
	I_{x}^{+} &= \dfrac{\lambda}{2} \exp \left( \dfrac{1}{2} \lambda^2 \mu^{r}_{j m} - \lambda \widehat{r}_{j m} \right) \Phi_{\mathcal{N}} \left( \tfrac{\widehat{r}^{+}_{j m}}{\sqrt{\mu^{r}_{j m}}} \right)
\end{align}
$ \widehat{r}_{j m}^{-} = \widehat{r}_{j m} + \lambda \mu_{j m}^{r} $, $ \widehat{r}_{j m}^{+} = \widehat{r}_{j m} - \lambda \mu_{j m}^{r} $, and $ \Phi_{\mathcal{N}}(x) = \frac{1}{\sqrt{2 \pi}} \int_{- \infty}^{x} \exp \left( - \frac{t^{2}}{2} \right) \mathrm{d} t $ is the cumulative distribution function of a standard normal distribution. The parameter $ \varpi_{j m} \in ( 0, 1 ) $ is viewed as the MRF module input, providing preceding support information. Apart from factor node $ \eta_{j' m}^{\alpha} $ and the two coupled factor nodes $ g_{j_{\mathrm{re}} m} $ and $ g_{j_{\mathrm{im}} m} $ with $ j_{\mathrm{re}} = j' $ and $ j_{\mathrm{im}} = j' + 2^{J} $, node $ b_{j' m} $ is also linked to its four neighboring support variable nodes. We mark the messages from the left, right, top, and bottom direction by $ \nu_{j' m}^{\mathrm{l}} $, $ \nu_{j' m}^{\mathrm{r}} $, $ \nu_{j' m}^{\mathrm{t}} $, and $ \nu_{j' m}^{\mathrm{b}} $, respectively. These messages can be calculated as
\begin{align}
	\nu_{j' m}^{d} = \kappa_{j' m}^{d} \delta ( b_{j' m} - 1 ) + \left( 1 - \kappa_{j' m}^{d} \right) \delta ( b_{j' m} + 1 ) \label{MRFmp}
\end{align}
where $ d \in \mathcal{D} = \{\mathrm{l}, \mathrm{r}, \mathrm{t}, \mathrm{b}\} $ and $ \kappa_{j' m}^{d} $ is given by
\begin{align}
	\kappa_{j' m}^{d} = \tfrac{\varpi_{j_{\mathrm{re}} m_{d}} \varpi_{j_{\mathrm{im}} m_{d}} \prod_{k \in \mathcal{D}_{d}} \kappa_{j' m_{d}}^{k} e^{- \alpha_{j'} + \beta_{j'}} + \left( 1 - \varpi_{j_{\mathrm{re}} m_{d}} \right) \left( 1 - \varpi_{j_{\mathrm{im}} m_{d}} \right) \prod_{k \in \mathcal{D}_{d}} \left( 1 - \kappa_{j' m_{d}}^{k} \right) e^{\alpha_{j'} - \beta_{j'}}}{\left( e^{\beta_{j'}} + e^{- \beta_{j'}} \right) \left( \varpi_{j_{\mathrm{re}} m_{d}} \varpi_{j_{\mathrm{im}} m_{d}} \prod_{k \in \mathcal{D}_{d}} \kappa_{j' m_{d}}^{k} e^{- \alpha_{j'}} + \left( 1 - \varpi_{j_{\mathrm{re}} m_{d}} \right) \left( 1 - \varpi_{j_{\mathrm{im}} m_{d}} \right) \prod_{k \in \mathcal{D}_{d}} \left( 1 - \kappa_{j' m_{d}}^{k} \right) e^{\alpha_{j'}} \right)}. \label{kappa}
\end{align}
In \eqref{kappa}, for instance, with respect to the left node $ b_{j' m_{\mathrm{l}}} $, $ \mathcal{D}_{\mathrm{l}} = \mathcal{D} \setminus \mathrm{r} = \{ \mathrm{l}, \mathrm{t}, \mathrm{b} \} $. Later, the backward message from $ b_{j' m} $ to $ g_{j m} $ is represented as
\setcounter{equation}{22}
\begin{align}
	\nu_{b_{j' m} \to g_{j m}} = \rho_{j m} \delta( b_{j' m} - 1 ) + ( 1 - \rho_{j m} ) \delta( b_{j' m} + 1 )
\end{align}
with
\begin{align}
	\rho_{j m} = \dfrac{\varpi_{q m} \prod_{d \in \mathcal{D}} \kappa_{j' m}^{d} e^{- \alpha_{j'}}}{\varpi_{q m} \prod_{d \in \mathcal{D}} \kappa_{j' m}^{d} e^{- \alpha_{j'}} + ( 1 - \varpi_{q m} ) \prod_{d \in \mathcal{D}} ( 1- \kappa_{j' m}^{d} ) e^{\alpha_{j'}}} \label{rho}
\end{align}
where the index $ q = j + 2^{J} $ when $ j \in \left[ 1, 2^{J} \right] $ and $ q = j - 2^{J} $ when $ j \in \left[ 2^{J} + 1, 2^{J + 1} \right] $. The parameter $ \rho_{j m} \in ( 0, 1 ) $ as the output of the MRF module offers estimated support information of $ x_{j m} $. Then, the message from $ g_{j m} $ to $ x_{j m} $ is a Bernoulli-Laplacian distribution expressed as
\begin{align}
	\nu_{g_{j m} \to x_{j m}} \propto \int_{b_{j' m}} p( x_{j m} | b_{j' m} ) \nu_{b_{j' m} \to g_{j m}} = \rho_{j m} \dfrac{\lambda}{2} \exp \left( - \lambda \left| x_{j m} \right| \right) + ( 1 - \rho_{j m} ) \delta(x_{j m}). \label{x}
\end{align}
As special cases, in Appendix B, we give examples of message updates of variable nodes in the edges/corners of the MRF structure.

We approximate the true marginal posterior $ p_{0}( x_{j m} | \mathbf{Y} ) $ by
\begin{align}
	p( x_{j m} | \mathbf{Y} ; \widehat{r}_{j m}, \mu^{r}_{j m}, \rho_{j m}, \lambda ) \propto \mathcal{N} ( x_{j m} ; \widehat{r}_{j m}, \mu^{r}_{j m} ) \cdot \nu_{g_{j m} \to x_{j m}} \label{marg}
\end{align}
using the aforementioned RLM inverse output Gaussian message and message $ \nu_{g_{j m} \to x_{j m}} $.
In Appendix B, we achieve closed forms of the marginal posterior mean and variance of $ x_{j m} $, in turn expressed as
\begin{align}
	\label{xhat}
	\widehat{x}_{j m} &= \rho_{j m} \tfrac{I_{x}^{-}}{I_{x}} \left[ \widehat{r}_{j m}^{-} - \mu_{j m}^{r} \tfrac{\mathcal{N} ( 0 ; \widehat{r}_{j m}^{-}, \mu_{j m}^{r} )}{\Phi_{\mathcal{N}} \left( -\widehat{r}_{j m}^{-} / \sqrt{\mu_{j m}^{r}} \right)} \right] + \rho_{j m} \tfrac{I_{x}^{+}}{I_{x}} \left[ \widehat{r}_{j m}^{+} + \mu_{j m}^{r} \tfrac{\mathcal{N} ( 0 ; \widehat{r}_{j m}^{+}, \mu_{j m}^{r} )}{\Phi_{\mathcal{N}} \left( \widehat{r}_{j m}^{+} / \sqrt{\mu_{j m}^{r}} \right)} \right] \\
	\label{xvar}
	\mu_{j m}^{x} &= \rho_{j m} \tfrac{I_{x}^{-}}{I_{x}} \left[ ( \widehat{r}_{j m}^{-} )^{2} \! + \! \mu_{j m}^{r} \! - \! \tfrac{\widehat{r}_{j m}^{-} \mu_{j m}^{r} \mathcal{N} ( 0 ; \widehat{r}_{j m}^{-}, \mu_{j m}^{r} )}{\Phi_{\mathcal{N}} \left( -\widehat{r}_{j m}^{-} / \sqrt{\mu_{j m}^{r}} \right)} \right] \! + \! \rho_{j m} \tfrac{I_{x}^{+}}{I_{x}} \left[ ( \widehat{r}_{j m}^{+} )^{2} \! + \! \mu_{j m}^{r} \! + \! \tfrac{\widehat{r}_{j m}^{+} \mu_{j m}^{r} \mathcal{N} ( 0 ; \widehat{r}_{j m}^{+}, \mu_{j m}^{r} )}{\Phi_{\mathcal{N}} \left( \widehat{r}_{j m}^{+} / \sqrt{\mu_{j m}^{r}} \right)} \right] \! - \! \widehat{x}_{j m}^{2}
\end{align}
where the normalization constant $ I_{x} $ is given by
\begin{align}
	I_{x} = \int_{x} \mathcal{N} ( x ; \widehat{r}_{j m}, \mu_{j m}^{r} ) \nu_{g_{j m} \to x_{j m}} = ( 1 - \rho_{j m} ) \mathcal{N} ( 0 ; \widehat{r}_{j m}, \mu_{j m}^{r} ) + \rho_{j m} \left( I_{x}^{-} + I_{x}^{+} \right). \label{Ix}
\end{align}

The aforementioned message components in Algorithm \ref{GAMP} are updated iteratively until a certain stopping criterion is satisfied. Apart from the limits on the maximum number of iterations, we leverage another normalized mean squared error (NMSE) based stopping criterion (see line 22 in \mbox{Algorithm \ref{GAMP}}) for certain tolerance $ \tau $. At last, the complex-valued estimation of $ \widetilde{\mathbf{X}} $, denoted by $ \underline{\mathbf{X}} $, can be easily obtained from the real-valued estimation $ \widehat{\mathbf{X}} $, i.e.,
\begin{align}
	\underline{\mathbf{X}} = \left[ \widehat{\mathbf{x}}_{1,:}^{T}, \dots, \widehat{\mathbf{x}}_{2^{J},:}^{T} \right]^{T} + \bar{i} \left[ \widehat{\mathbf{x}}_{2^{J} + 1,:}^{T}, \dots, \widehat{\mathbf{x}}_{2^{J + 1},:}^{T} \right]^{T}
\end{align}
where $ \bar{i} = \sqrt{-1} $.

Finally, to learn the activity pattern of codewords, we make a hard decision on the support of codewords with an appropriate threshold $ \upsilon $:
\begin{align}
	\mathcal{X} = \left\lbrace i : \left\| \underline{\mathbf{x}}_{i, :} \right\|^{2} > \upsilon, i \in \left[ 2^{J} \right] \right\rbrace \label{L}
\end{align}
where $ \underline{\mathbf{x}}_{i, :} $ is the $ i $-th row of $ \underline{\mathbf{X}} $. Note that the length of $ \mathcal{X} $ is not obligated to be $ K_{\mathrm{a}} $ since two or more users may select the same codeword to send at the same transmission slot. Recall that in the URA scenario, the BS has no obligation to discern any active user identity, thus we do not seek to reconstruct the codeword selection matrix $ \mathbf{\Xi} $ in \eqref{sigmod}.

\subsection{Parameter Learning via Expectation Maximization}

Note that some parameters required by the iterative process of GAMP, including noise variance $ \sigma^{2} $ and Laplace rate $ \lambda $, are typically unknown to the detection side. Denote by $ \boldsymbol{\theta} = [ \sigma^{2}, \lambda ]^{T} $ the complete vector of unknown parameters. Our purpose is to find the ML estimate $ \widehat{\boldsymbol{\theta}} $ of $ \boldsymbol{\theta} $ from the received signal $ \mathbf{Y} $, i.e., $ \widehat{\boldsymbol{\theta}} = \arg \operatorname*{max} \limits_{\boldsymbol{\theta}} \ln p( \mathbf{Y} ; \boldsymbol{\theta} ) $. The EM algorithm gives the solution to $ \widehat{\boldsymbol{\theta}} $ recursively taking the following two steps (detailed deduction can be found in \cite{VS13}):
\begin{itemize}[leftmargin=*]
	\item
	\textbf{Expectation Step (E-STEP):} Replace $ \ln p( \mathbf{Y} ; \boldsymbol{\theta} ) $ by the conditional expectation with respect to $ p( \mathbf{X} | \mathbf{Y} ; \widehat{\boldsymbol{\theta}}( t ) ) $:
	\begin{align}
		\! \mathbb{E} \left\lbrace \ln p( \mathbf{Y}, \mathbf{X} ; \boldsymbol{\theta} ) \right\rbrace = \int_{\mathbf{X}} p( \mathbf{X} | \mathbf{Y} ; \widehat{\boldsymbol{\theta}}( t ) ) \ln p( \mathbf{Y}, \mathbf{X} ; \boldsymbol{\theta} ).
	\end{align}
	\item
	\textbf{Maximization Step (M-STEP):} Maximize the above average log-likelihood:
	\begin{align}
		\widehat{\boldsymbol{\theta}}( t + 1 ) = \arg \operatorname*{max} \limits_{\boldsymbol{\theta}} \ \mathbb{E} \left\lbrace \ln p( \mathbf{Y}, \mathbf{X} ; \boldsymbol{\theta} ) \right\rbrace.
	\end{align}
\end{itemize}

For convenience, we divide the overall ML estimation problem into two tractable parts, each independently solved by an EM algorithm. These algorithms manifest as the recursions of the following optimization problems \cite{VS13}
\begin{align}
	\label{EMsigma}
	\theta_{\sigma^2}( t + 1) &= \arg \operatorname*{max} \limits_{\sigma^2} \ \sum_{n} \sum_{m} \mathbb{E} \left\lbrace \ln p( y_{n m} | z_{n m} ; \sigma^{2} ) \right\rbrace
\end{align}
where the expectation is taken over $ p( z_{n m} | \mathbf{Y} ; \boldsymbol{\theta} ) $, and
\begin{align}
	\label{EMlambda}
	\theta_{\lambda}( t + 1 ) &= \arg \operatorname*{max} \limits_{\lambda} \ \sum_{j} \sum_{m} \mathbb{E} \left\lbrace \ln p( x_{j m} ; \lambda ) \right\rbrace
\end{align}
where the expectation is taken over $ p( x_{j m} | \mathbf{Y} ; \boldsymbol{\theta} ) $. Note that alternately solving \eqref{EMsigma} and \eqref{EMlambda} may not lead to the optimal $ \widehat{\boldsymbol{\theta}} $, but it is more computationally tractable and helps with the convergence.

We first derive the EM update for the noise variance $ \sigma^{2} $. The maximizing value of $ \sigma^{2} $ in \eqref{EMsigma} is certainly the value of $ \sigma^{2} $ when the derivative of the sum equals to zero, i.e.,
\begin{align}
	\sum_{n} \sum_{m} \int_{z_{n m}} p( z_{n m} | \mathbf{Y} ; \boldsymbol{\theta} ) \dfrac{\mathrm{d}}{\mathrm{d} \sigma^{2}} \ln p( y_{n m} | z_{n m} ; \sigma^{2} ) = 0. \label{deri_sigma1}
\end{align}
With $ p( y_{n m} | z_{n m} ; \sigma^{2} ) = \mathcal{N}( y_{n m} ; z_{n m}, \sigma^{2} ) $, we have
\begin{align}
	\dfrac{\mathrm{d}}{\mathrm{d} \sigma^{2}} \ln p( y_{n m} | z_{n m} ; \sigma^{2} ) = \dfrac{1}{2 \sigma^{2}} \left[ \dfrac{( y_{n m} - z_{n m} )^{2}}{\sigma^{2}} - 1 \right]. \label{deri_sigma2}
\end{align}
By plugging \eqref{deri_sigma2} into \eqref{deri_sigma1}, we obtain the unique solution to \eqref{EMsigma} expressed as
\begin{align}
	\sigma^{2} &= \dfrac{1}{2 N M} \sum_{n} \sum_{m} \int_{z_{n m}} ( y_{n m} - z_{n m} )^{2} p( z_{n m} | \mathbf{Y} ; \boldsymbol{\theta} ) \notag \\
	&= \dfrac{1}{2 N M} \sum_{n} \sum_{m} \left[ ( y_{n m} - \widehat{z}_{n m} )^{2} + \mu_{n m}^{z} \right]. \label{sigma}
\end{align}

Then, similar processes can be implemented to learn the Laplace rate $ \lambda $. With the distribution of the component $ x_{j m} $ given in \eqref{x}, it is readily seen that
\begin{align}
	\dfrac{\mathrm{d}}{\mathrm{d} \lambda} \ln p( x_{j m} ; \lambda ) &= \dfrac{\frac{\rho_{j m}}{2} ( 1- \lambda \left| x_{j m} \right| ) \exp \left( - \lambda \left| x_{j m} \right| \right)}{\rho_{j m} \frac{\lambda}{2} \exp \left( - \lambda \left| x_{j m} \right| \right) + ( 1 - \rho_{j m} ) \delta(x_{j m})}
	= \begin{cases}
		0, & x_{j m} = 0 \\
		\tfrac{1}{\lambda} - \left| x_{j m} \right|, & x_{j m} \neq 0
	\end{cases}.
\end{align}
The derived function above is not continuous, thus, we define the closed ball $ \mathcal{X}_{\epsilon} = [ -\epsilon, \epsilon ] $ and its complementary set over the real number field $ \overline{\mathcal{X}}_{\epsilon} = \mathbb{R} \backslash \mathcal{X}_{\epsilon} $ to describe the field of integration. When $ \epsilon \to 0 $, the derivative of the sum conditional expectation in \eqref{EMlambda} can be computed as
\begin{align}
	&\sum_{j} \sum_{m} \int_{x_{j m}} p( x_{j m} | \mathbf{Y} ; \boldsymbol{\theta} ) \dfrac{\mathrm{d}}{\mathrm{d} \lambda} \ln p( x_{j m} ; \lambda ) \notag \\
	&\quad = \sum_{j} \sum_{m} \dfrac{\rho_{j m}}{\lambda} - \sum_{j} \sum_{m} \lim\limits_{\epsilon \to 0} \int_{x_{j m} \in \overline{\mathcal{X}}_{\epsilon}} \left| x_{j m} \right| p( x_{j m} | \mathbf{Y} ; \boldsymbol{\theta} ). \label{deri_lambda}
\end{align}
By setting \eqref{deri_lambda} to be zero, we have the EM update for the scale parameter $ \lambda $ expressed as
\begin{align}
	\lambda = \dfrac{\sum_{j} \sum_{m} \rho_{j m}}{\sum_{j} \sum_{m} \tfrac{\rho_{j m}}{I_{x}} \left\lbrace I_{x}^{+} \left[ \widehat{r}_{j m}^{+} + \mu_{j m}^{r} \tfrac{\mathcal{N} ( 0 ; \widehat{r}_{j m}^{+}, \mu_{j m}^{r} )}{\Phi_{\mathcal{N}} \left( \widehat{r}_{j m}^{+} / \sqrt{\mu_{j m}^{r}} \right)} \right] - I_{x}^{-} \left[ \widehat{r}_{j m}^{-} - \mu_{j m}^{r} \tfrac{\mathcal{N} ( 0 ; \widehat{r}_{j m}^{-}, \mu_{j m}^{r} )}{\Phi_{\mathcal{N}} \left( -\widehat{r}_{j m}^{-} / \sqrt{\mu_{j m}^{r}} \right)} \right] \right\rbrace}. \label{lambda}
\end{align}

For EM initialization, we set the Laplacian rate $ \lambda = 1 $ and the noise variance $ \sigma^{2} = \frac{\| \mathbf{Y} \|_{F}^{2}}{2 M N ( R + 1 )} $ with $ R $ the overall signal-to-noise ratio (SNR) defined by $ \mathbb{E}\{\left\| \mathbf{X} \right\|^{2}_{F}\} / \mathbb{E}\{\left\| \mathbf{W} \right\|^{2}_{F}\} $. When the true SNR is unknown, $ R = 100 $ is recommended \cite{VS13}. The EM framework can also be adapted to study MRF parameters, while we directly set $ \alpha_{j'} = \beta_{j'} = 0.4 $ as suggested in \cite{SS11}.

\subsection{Performance Analysis}

\subsubsection{Asymptotic Analysis}

It is well known that the AMP/GAMP algorithm can be analyzed by \emph{state evolution} (SE) \cite{BM11} in the asymptotic area where $ N,2^{J} \to \infty $ while their ratio converges to a fixed positive value $ \delta = 2^{J}/N $. Viewing the output $ \underline{\mathbf{X}} $ of EM-MRF-GAMP as a signal plus Gaussian noise, SE provides a scalar equivalent model for the mean square error performance prediction of the algorithm. Define a set of random vectors $ \widehat{\boldsymbol{x}}_{i}(t) = \boldsymbol{x}_{i} + \varrho_{i}(t) \boldsymbol{v}_{i} $, $ i \in [2^{J}] $, where $ \boldsymbol{x}_{i} \in \mathbb{C}^{M} $ captures the distribution of $ \widetilde{\mathbf{x}}_{i,:}^{T} $, $ \boldsymbol{v}_{i} \in \mathbb{C}^{M} \sim \mathcal{CN}(0,\mathbf{I}_{M}) $, and $ \varrho_{i} $ known as the \emph{state} is iteratively computed as \cite{BMN20}
\begin{align}
	\varrho^{2}(t+1) = 2\sigma^{2} + \delta \mathbb{E}\left\lbrace \left\| \eta_{\mathrm{de}}(\boldsymbol{x} + \varrho(t) \boldsymbol{v}) - \boldsymbol{x} \right\|^{2} \right\rbrace
\end{align}
where $ \eta_{\mathrm{de}}(\cdot) $ is the denoiser. Note that due to the underlying structured channel sparsity captured by MRF, the denoiser of EM-MRF-GAMP is non-separable \cite{MRB17}. Leveraging SE, we have the following proposition.

\begin{pro}
	Suppose that $ \boldsymbol{x} $ captures the distribution of $ \widetilde{\mathbf{x}}_{\cdot,:}^{T} $ in \eqref{sigmod} and $ \boldsymbol{v} \sim \mathcal{CN}(0, \mathbf{I}_{M}) $, the likelihood of $ \widehat{\boldsymbol{x}} = \boldsymbol{x} + \varrho \boldsymbol{v} $ given $ \boldsymbol{x} = \mathbf{0} $ is expressed as
	\begin{align}
		p(\widehat{\boldsymbol{x}} | \boldsymbol{x} = \mathbf{0}) = \dfrac{\exp \left( - \left\| \widehat{\boldsymbol{x}} \right\|^{2} \varrho^{-2} \right)}{\pi^{M} \varrho^{2M}}. \label{pro1}
	\end{align}
\end{pro}

\begin{proof}
	Given $ \boldsymbol{x} = \mathbf{0} $, we have $ \widehat{\boldsymbol{x}} = \varrho \boldsymbol{v} \sim \mathcal{CN}(0, \varrho^{2} \mathbf{I}_{M}) $, leading to \eqref{pro1}.
\end{proof}

Now we evaluate the detection performance of EM-MRF-GAMP using the criterion of per-user probability of error (PUPE) in \cite{P07} defined as $ \text{PUPE} \triangleq \mathbb{E}\{ | \mathcal{A}_{\mathrm{a}} \backslash \mathcal{X} | / |\mathcal{A}_{\mathrm{a}}| \} $, where $ \mathcal{A}_{\mathrm{a}} $ is the set of indexes of active codewords, and $ \mathcal{X} $ is defined in \eqref{L}. For convenience but without loss of generality, we follow the assumption in \cite{P07,FHJ21} that exactly $ K_{\mathrm{a}} = |\mathcal{A}_{\mathrm{a}}| $ active codewords are determined active without codeword collisions (since $ 2^{J} \to \infty $). Thus, we have $ \mathbb{E}\{ | \mathcal{A}_{\mathrm{a}} \backslash \mathcal{X} | / |\mathcal{A}_{\mathrm{a}}| \} = \mathbb{E}\{ | \mathcal{X} \backslash \mathcal{A}_{\mathrm{a}} | / |\mathcal{X}| \} $, leading to the following corollary.

\begin{cor}
	Given $ \mathcal{A}_{\mathrm{a}} $ the set of indexes of active codewords, and $ \mathcal{X} $ in \eqref{L}. With a threshold $ \upsilon > 0 $, $ \mathrm{PUPE} = \mathbb{E}\{ | \mathcal{A}_{\mathrm{a}} \backslash \mathcal{X} | / |\mathcal{A}_{\mathrm{a}}| \} $ can be computed as
	\begin{align}
		\mathrm{PUPE} = \int_{\left\| \widehat{\boldsymbol{x}} \right\|^{2} > \upsilon} p(\widehat{\boldsymbol{x}} | \boldsymbol{x} = \mathbf{0}) d\widehat{\boldsymbol{x}} = \dfrac{\overline{\Gamma}(M, \upsilon \varrho^{-2})}{\Gamma(M)}
	\end{align}
	where $ \Gamma(\cdot) $ and $ \overline{\Gamma}(\cdot,\cdot) $ denote the Gamma function and the upper incomplete Gamma function, respectively. Further, suppose that the threshold $ \upsilon = c \mathbb{E}\{\left\| \boldsymbol{v} \right\|^{2} \} = c M \varrho^{2} $ with $ c > 1 $, we have
	\begin{align}
		\lim\limits_{M \to \infty} \dfrac{\overline{\Gamma}(M, \upsilon \varrho^{-2})}{\Gamma(M)} = 0.
	\end{align}
\end{cor}

\begin{proof}
	See Appendix C.
\end{proof}

Corollary 1 claims that with an appropriate threshold setting, the detection error rate of EM-MRF-GAMP tends to be zero when the number of antennas grows to infinity, revealing the benefit of massive MIMO. Note that to provide valuable insights of asymptotic AD performance of the proposed algorithm, we do not consider explicit expression or specific structure of the denoiser. Consequently, how the state $ \varrho $ evolves with the number of iterations is not addressed in the above analysis. Rigorous asymptotic analyses dealing with this issue for AMP with non-separable denoisers can be found in \cite{BMN20,MRB17}, where more strict assumptions are taken compared to those considered in this paper.


\subsubsection{Computational Complexity Analysis}

The computations for lines 6-9 in Algorithm \ref{GAMP} and those for lines 13, 17, and 18-19 yield the complexity of $ \mathcal{O}( N M ) $ and $ \mathcal{O}( 2^{J} M ) $, respectively. The calculations related to the MRF estimation module in line 15 have the complexity of $ \mathcal{O}( T_{\mathrm{mrf}} 2^{J} M ) $, where the number of iterations $ T_{\mathrm{mrf}} $ is relatively small and has limited effects to the overall complexity. The EM updates of $ \sigma^{2} $ and $ \lambda $ are computed in $ \mathcal{O}( N M ) $ and $ \mathcal{O}( 2^{J} M ) $ times, respectively. As $ 2^{J} $ grows, the most of the computing resources are contributed to the matrix multiplications in lines 4-5 and 10-11, each requiring $ 2^{J} N M $ multiplications. In general, the complexity order of the proposed algorithm per iteration is $ \mathcal{O}( 2^{J} N M ) $, which is on the same level as other message passing based CS algorithms like MMV-AMP \cite{LY18} and GAMP \cite{R11}. Since the computational complexity increases linearly with $ M $, the proposed EM-MRF-GAMP algorithm is computationally efficient in the massive MIMO setting.

\section{Proposed Clustering Algorithm for Clustering-Based Decoding}

After retrieving active codewords and their corresponding channels from all slots, the BS reconstructs the original message list by distinguishing slot-distributed channels of each active user in a clustering way. In this section, we provide a modified constrained clustering algorithm tailored for message stitching with a refinement to restrain the impact of codeword collision.

\subsection{Slot-Balanced $ K $-means for Constrained Clustering}

We provisionally consider an ideal circumstance where there are no users selecting the same codeword at the same time, such that exactly $ K_{\mathrm{a}} $ codewords are judged to be active in every slot. The clustering decoder aims to sort the associated channels into $ K_{\mathrm{a}} $ groups according to some notions of similarity, and obtain each message based on the permutation of codewords. As a well-known approach for data classification, $ K $-means clustering \cite{HW79} automatically partitions a data set into groups with low intra-group distances and high inter-group distances. It proceeds by choosing $ K $ random group centers as the initializer, and then iteratively amending them taking the following two steps:
\begin{itemize}[leftmargin=*]
	\item
	\textbf{Assignment Step:} Each data instance is assigned to the closest cluster center.
	\item
	\textbf{Update Step:} Each cluster center is updated to be the centroid of its constituent data instances.
\end{itemize}
These steps are repeated until there are no further changes in centroid locations. For convenience, we denote the reconstructed channels of active codewords at the $ s $-th slot by $ \mathbf{G}_{s} = [ ( \underline{\mathbf{x}}_{i_{1}, :}^{s} )^{T}, \dots, ( \underline{\mathbf{x}}_{i_{K_{\mathrm{a}}}, :}^{s} )^{T} ] \in \mathbb{C}^{M \times K_{\mathrm{a}}} $, $ i_k \in \mathcal{X}_{s} $. To find the main lobe of the angular domain channel obtained by a DFT transformation (see Appendix A), we take the absolute value of $ \mathbf{G}_{s} $ and construct the data set to be classified as $ \mathcal{R} = \{ \mathbf{R}_{s}: s \in [ S ] \} $ with $ \mathbf{R}_{s} = [ \mathbf{r}_{1}^{s}, \dots, \mathbf{r}_{K_{a}}^{s} ]^{T} = | \mathbf{G}_{s} | $. The center points (centroids) of $ K_{\mathrm{a}} $ groups are represented by $ \mathbf{C} = [ \mathbf{c}_{1}, \dots, \mathbf{c}_{K_{\mathrm{a}}} ]^{T} $.

Traditional $ K $-means algorithm set no limitation conditions when classifying data. However, in the application scene of message stitching, the decoder is mandatory to satisfy two obvious constraints \cite{SBM21}:
\begin{itemize}[leftmargin=*]
	\item
	\textbf{Constraint I:} Channels from the same slot can not be allocated to the same group.
	\item
	\textbf{Constraint II:} Each group must consist of $ S $ channels at the end of the clustering.
\end{itemize}
To proceed as in $ K $-means with Constraint II, we perform the assignment step on a per-slot basis, i.e., all $ K_{\mathrm{a}} $ channels obtained from the same slot are allocated to $ K_{\mathrm{a}} $ groups in one step. At each assignment step, to meet Constraint I, we tend to solve the following assignment problem:
\begin{subequations}
	\begin{align}
		\label{dis}
		\mathop{\operatorname{minimize}}\limits_{\mathbf{\Gamma}} & \ \sum\nolimits_{k = 1}^{K_{\mathrm{a}}} \sum\nolimits_{k' = 1}^{K_{\mathrm{a}}} \gamma_{k, k'} d( k, k' ) \\
		\operatorname{subject \ to} & \ \sum\nolimits_{k' = 1}^{K_{\mathrm{a}}} \gamma_{k, k'} = 1, \forall k \in [K_{\mathrm{a}}] \\
		& \ \sum\nolimits_{k = 1}^{K_{\mathrm{a}}} \gamma_{k, k'} = 1, \forall k' \in [K_{\mathrm{a}}] \\
		& \ \gamma_{k, k'} \in \{ 0, 1 \}, k \in [K_{\mathrm{a}}], k' \in [K_{\mathrm{a}}].
	\end{align}
\end{subequations}
We measure the distance between a channel vector and a group center by the Euclidean distance, i.e., $ d( k, k' ) = \left\| \mathbf{r}_{k} - \mathbf{c}_{k'} \right\| $ in \eqref{dis}. As the solution to the above linear programming problem, $ \mathbf{\Gamma} \in \{0, 1\}^{K_{\mathrm{a}} \times K_{\mathrm{a}}} $ is a binary matrix whose $ ( k, k' ) $-th entry $ \gamma_{k, k'} = 1 $ indicates that the $ k $-th channel belongs to the $ k' $-th group. We appeal to the famous \emph{Hungarian algorithm} \cite{K55} to get the optimal assignment. As the algorithm input, weights are stored in a cost matrix $ \mathbf{D} \in \mathbb{R}^{K_{\mathrm{a}} \times K_{\mathrm{a}}} $ with the $ ( k, k' ) $-th entry calculating the Euclidean distance between $ \mathbf{r}_{k} $ and $ \mathbf{c}_{k'} $. After grouping according to the algorithm output $ \mathbf{\Gamma} $, the update step is similar to that of $ K $-means, where each new group center is calculated as the mean of the constituent channel vectors.

We name the proposed clustering algorithm \emph{slot-balanced $ K $-means} since it performs assignment slot by slot and obtains clusters with identical numbers of constituting elements. Obviously, it requires $ S $ assignment steps to finish one round of channel partitioning concerning all transmission slots. Denote $ \mathbf{c}_{k'}^{t, s} $ the updated centroid of the $ k' $-th group at round $ t $, step $ s $. With the assignment matrix $ \mathbf{\Gamma} $ acquired from the $ s $-th assignment step, the update step goes by
\begin{align}
	\mathbf{c}_{k'}^{t, s} = \dfrac{1}{s} \left[ (s - 1) \mathbf{c}_{k'}^{t, s - 1} + \sum_{k = 1}^{K_{\mathrm{a}}} \gamma_{k, k'} \mathbf{r}_{k}^{s} \right]. \label{c1}
\end{align}
The initial centroids of each round are inherited from the final renewed center points of the former round, i.e., $ \mathbf{C}^{t, 0} = \mathbf{C}^{t - 1, S} $. As the initialization of the algorithm, $ \mathbf{C}^{0, S} $ can be generated randomly or set to be $ \mathbf{R}_{s} $ chosen randomly from any slot.

\subsection{Codeword Collision Resolution}

\begin{algorithm}[!t]
	\caption{Slot-Balanced $ K $-means for Clustering Decoding}
	\label{Kmeans}
	\begin{algorithmic}[1]
		\STATE \textbf{Input:} Data set $ \{ \mathbf{r}_{s, k} \in \mathbb{C}^{M} : s \in [ S ], k \in [ K_{s} ] \} $, maximum number of iterations $ T_{\mathrm{c}} $
		\STATE \textbf{Initialize:} Centroid locations $ \mathbf{C}^{0, S} = [ \mathbf{c}_{1}^{0, S}, \dots, \mathbf{c}_{K_{\mathrm{a}}}^{0, S} ] $
		\FOR{$ t = 1, 2, \dots, T_{\mathrm{c}} $}
		\STATE Set $ \mathbf{C}^{t, 0} = \mathbf{C}^{t - 1, S} $
		\FOR{$ s = 1, 2, \dots, S $}
		\STATE $ \forall k \in [ K_{s} ], k' \in [ K_{\mathrm{a}} ] $: Compute cost matrix $ \mathbf{D} $ with $ d( k, k' ) = \left\| \mathbf{r}_{k}^{s} - \mathbf{c}_{k'}^{t, s - 1} \right\| $
		\IF{$ K_{s} < K_{\mathrm{a}} $}
		\STATE Add $ K_{\mathrm{a}} - K_{s} $ rows with the largest sum of elements to the matrix $ \mathbf{D} $
		\ENDIF
		\STATE Solve the assignment problem \eqref{dis} by the Hungarian algorithm with output $ \mathbf{\Gamma} $
		\IF{$ K_{s} = K_{\mathrm{a}} $}
		\STATE $ \forall k' $: Update $ \mathbf{c}_{k'}^{t, s} $ via \eqref{c1}.
		\ELSE
		\STATE $ \forall k' $: Update $ \mathbf{c}_{k'}^{t, s} $ via \eqref{c2}.
		\ENDIF
		\ENDFOR
		\STATE \textbf{if} $ \mathbf{C}^{t, S} = \mathbf{C}^{t - 1, S} $, \textbf{stop}
		\ENDFOR
		\STATE \textbf{Output:} Partitioning of the data set
	\end{algorithmic}
\end{algorithm}

Codeword collisions are unavoidable to appear when a large number of active users share a common codebook with limited codewords. If at least two users choose the same codeword to send simultaneously, the equivalent channel of the reused codeword is the sum of their corresponding channels (see \eqref{xcolumn}). Fortunately, provided that these confronted users are geographically separated, their broadcast signals will undergo different scatterers with different AoA intervals to the BS. Therefore, the sparse channel based on information recovered from other slots where such a user is not involved in any codeword collisions.

In the case of codeword collision, the proposed slot-balanced $ K $-means can still work with a bit of adjustment. The number of active users is first judged to be $ K_{\mathrm{a}} = \operatorname{max} \{ K_{s}, s \in [ S ] \} $. Codeword reuse is deduced to happen at the $ s $-th slot when $ K_{s} < K_{\mathrm{a}} $. The cost matrix $ \mathbf{D} $ is first computed as a $ K_{s} \times K_{\mathrm{a}} $-dimensional matrix. Since the Hungarian algorithm only operates with a square matrix, we select $ K_{\mathrm{a}} - K_{s} $ rows of $ \mathbf{D} $ with the largest sum of elements and append them to $ \mathbf{D} $ to form a $ K_{\mathrm{a}} \times K_{\mathrm{a}} $-dimensional input matrix. The channel corresponding to each duplicated row is allocated to more than one group by the Hungarian algorithm. However, such a contaminated channel vector should not be straightly used to calculate the next centroid. We leverage the unique angular transmission pattern revealed by the center point of each cluster to counteract the interference of other conflicting users, with the update step expressed as
\begin{align}
	\mathbf{c}_{k'}^{t, s} = \dfrac{1}{s} \left[ (s - 1) \mathbf{c}_{k'}^{t, s - 1} + \sum_{k = 1}^{K_{\mathrm{a}}} \gamma_{k, k'} \mathbf{\Lambda}_{k'}^{t, s - 1} \mathbf{r}_{k}^{s} \right]. \label{c2}
\end{align}
where $ \mathbf{\Lambda}_{k'}^{s - 1} \in \{0, 1\}^{M \times M} $ is a diagonal matrix with indexes of none-zero diagonal elements denoted by $ \mathcal{A} $. The set $ \mathcal{A} $ is chosen such that the elements $ \{ c_{k' m}^{t, s - 1}: m \in \mathcal{A} \} $ concentrate most of the energy of the vector $ \mathbf{c}_{k'}^{t, s - 1} $, i.e., $ \sum_{m \in \mathcal{A}} \left| c_{k' m}^{t, s - 1} \right|^{2} > \zeta \left\| \mathbf{c}_{k'}^{t, s - 1} \right\|^{2} $ for a given threshold $ \zeta $ (e.g. $ \zeta = 0.95 $).

\subsection{Further Discussions}

We summarize the overall algorithm in Algorithm \ref{Kmeans}. Our method is a special case of the constrained $ K $-means \cite{WCR01} where channels recovered at each slot formulate couples of cannot-link constraints with each other. Same as the constrained $ K $-means, the proposed iterative clustering algorithm is guaranteed to converge. Note that even though the data assignment step and centroid update step are both optimal, the final solution often reaches a local optimum. Our method can also be treated as a revision of the balanced $ K $-means \cite{MF14} to satisfy Constraint I. Dominated by the Hungarian algorithm computed in $ \mathcal{O}( K_{\mathrm{a}}^{3} ) $ time, the algorithm complexity yields the order of $ \mathcal{O}( S K_{\mathrm{a}}^{3} ) $, which vastly outperforms the constrained $ K $-means of complexity $ \mathcal{O}( S^{3.5} K_{\mathrm{a}}^{7} ) $.

\section{Simulation Results}

In this section, we conduct numerical experiments to evaluate the performance of the proposed UCS scheme. We consider a circumstance where $ K_{\mathrm{a}} = 100 $ active users are randomly and uniformly located in a semicircular coverage area with a radius of $ 50 $ meters, while the value of $ K_{\mathrm{a}} $ is unknown to the decoder. We would like to mention that the system does not acquire the knowledge of $ K_{\mathrm{tot}} $. The number of inactive users within the URA model can be arbitrarily large, but the system performance depends only on $ K_{\mathrm{a}} $.

We generate the virtual MIMO channel by a general 3D wireless channel model \cite{WWA18}. Such a geometry-based stochastic model (GBSM) is derived from the predefined stochastic distributions of effective scatterers by applying the fundamental laws of wave propagation. We consider a non-line-of-sight (NLOS) propagation environment. There are $ 16 $ random scatterers each with $ 7.0 $ degrees angular spread in azimuth and $ 19.0 $ degrees angular spread in elevation \cite{FLE20}; they are randomly effective for a given active user. The carrier frequency is $ 2.6 $ GHz, and other parameters related to scatterers are set according to \cite[Table II]{FLE20}. Since we consider a block-fading narrow-band MIMO channel in this paper, we treat parameters in \cite[Table I]{WWA18} as time-invariant, i.e., we do not consider scenes like target movement, array-time cluster evolution, and mean power updates of rays specific to the model in \cite{WWA18}. Given the coordinates of transmitting/receiving antennas and the statistics of scatterers, the spatial domain channel $ \widetilde{\mathbf{H}}_{k} $ can be easily generated. Since GBSM captures the characteristic that MIMO channels propagate in the form of clusters of paths, the transformed angular domain channel $ \mathbf{H}_{k} $ exhibits the clustered sparsity structure as shown in Fig. \ref{channel}.

\subsection{Performance of EM-MRF-GAMP Algorithm}

We choose the measurement matrix $ \widetilde{\mathbf{A}} \in \mathbb{C}^{N \times 2^{12}} $ as an i.i.d Gaussian matrix with $ N $ the number of measurements. As the stopping criteria for the iterative algorithm, we set $ T_{\mathrm{max}} = 50 $, $ T_{\mathrm{mrf}} = 20 $, and the precision tolerance $ \tau = 10^{-5} $. We assess the algorithm in the aspect of CE accuracy by the NMSE of the recovered active channels, i.e., $ \text{NMSE} = \| \mathbf{X}_{\mathrm{o}} - \underline{\mathbf{X}} \|_{F}^{2} / \| \mathbf{X}_{\mathrm{o}} \|_{F}^{2} $ with $ \mathbf{X}_{\mathrm{o}} $ the original channel matrix arranged by the indexes in $ \mathcal{X} $ (c.f. \eqref{L}). Note that the AD performance of EM-MRF-GAMP is integrated into the systematic error for consideration.

We consider several algorithms from the Bayesian family for comparison: 1) \textbf{GAMP-Laplace} \cite{BSY19}\textbf{:} a GAMP-based algorithm with a Bernoulli-Laplacian prior on $ x $; 2) \textbf{MMV-GAMP:} based on \cite{BSY19}, the multiple measurement vector (MMV) setting \cite{ZS13} is introduced to capture the row sparsity of $ \mathbf{X} $; 3) \textbf{CB-CS+LMMSE} \cite{FJC21'2}\textbf{:} the covariance-based CS (CB-CS) estimator in \cite{FHJ19} first reconstructs the LSFCs of channels for AD, and the linear minimum mean-square error (LMMSE) estimation is then performed on active codewords for CE. We depict the average NMSE performance of the aforementioned methods versus the SNR in \mbox{Fig. \ref{NMSEa}}. It can be seen that the proposed EM-MRF-GAMP algorithm significantly outperforms other approaches since it well captures the clustered support structure of the sparse angular domain channel. \mbox{Fig. \ref{NMSEb}} exhibits the NMSE performance as a function of the number of measurements (i.e, coherent block-length). For instance, at target $ \text{NMSE} = -20\text{dB} $ when $ \text{SNR} = 10 \text{dB} $, the EM-MRF-GAMP algorithm requires about $ 120 $ measurements, while MMV-GAMP needs more than $ 160 $ measurements. As CE is key to the clustering decoder for message stitching, in order to reach the same level of decoding error probability, the spectral efficiency of the UCS scheme where EM-MRF-GAMP acts as the CS decoder is certainly higher than that of the UCS scheme with the MMV-GAMP based CS decoder.

\begin{figure}[!t]
	\centering
	\subfigure[]{\includegraphics[height=5.5cm]{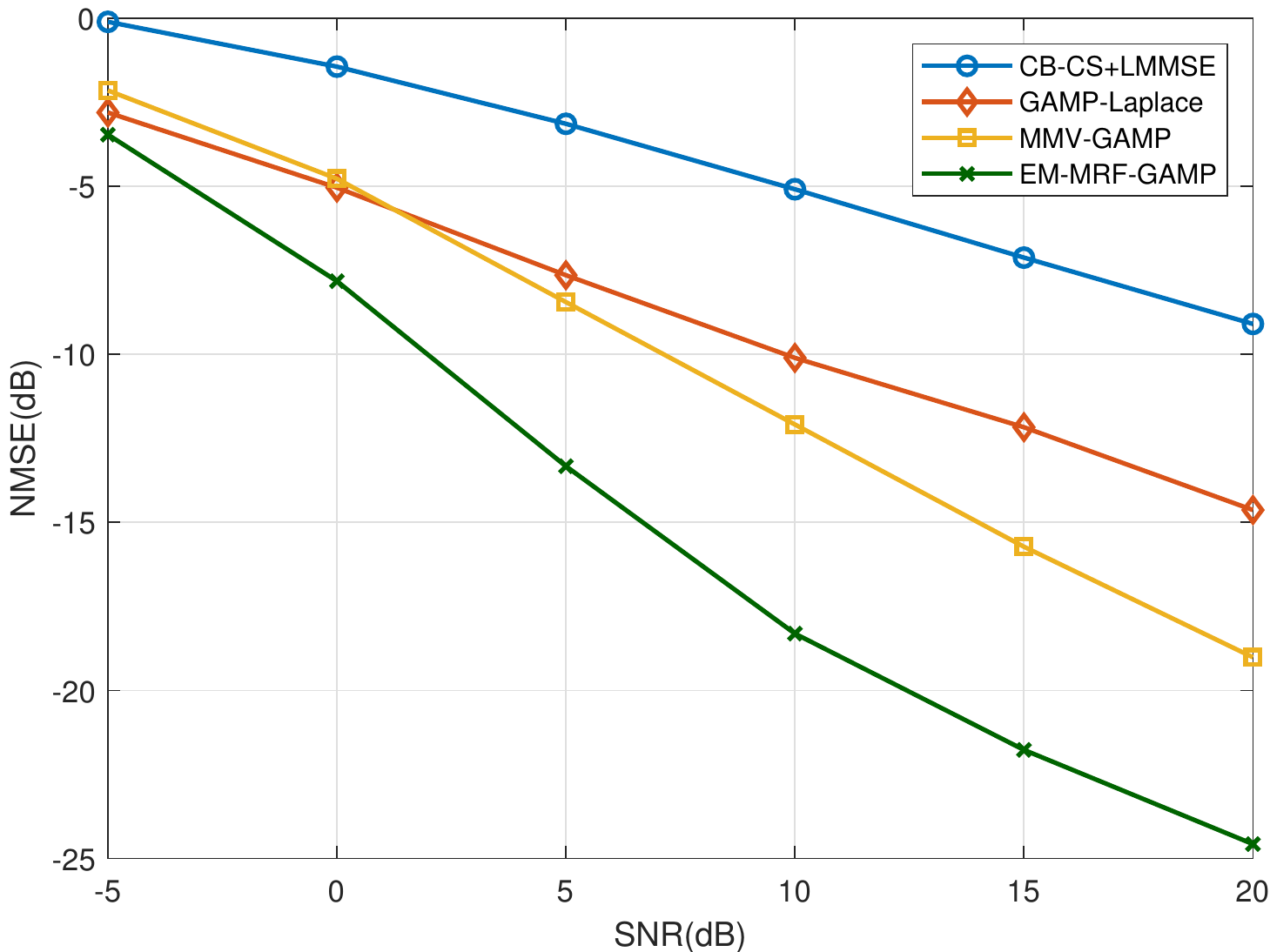}
	\label{NMSEa}}
	\hspace{0.15in}
	\subfigure[]{\includegraphics[height=5.5cm]{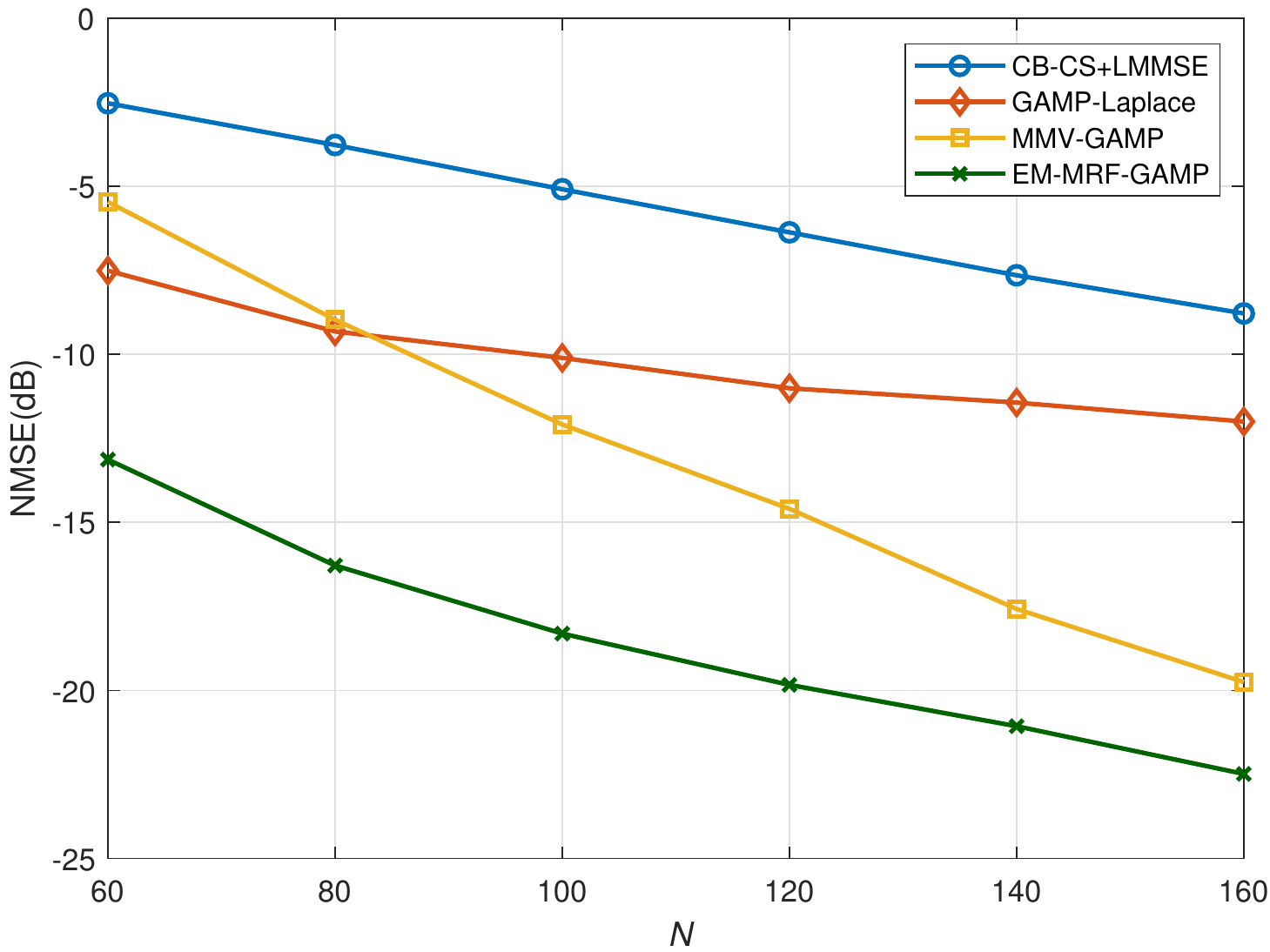}
	\label{NMSEb}}
	\caption{The NMSEs of various algorithms. $ K_{\mathrm{a}} = 100 $, $ M_{\mathrm{v}} = 4 $, $ M_{\mathrm{h}} = 25 $ (i.e., $ M = 4 \times 25 = 100 $). a) NMSEs versus the SNR when $ N = 100 $. b) NMSEs versus the number of measurements when $ \text{SNR} = 10 \text{dB} $.}
	\label{NMSE}
\end{figure}

\begin{figure}[!t]
	\centering
	\subfigure[]{\includegraphics[height=5.5cm]{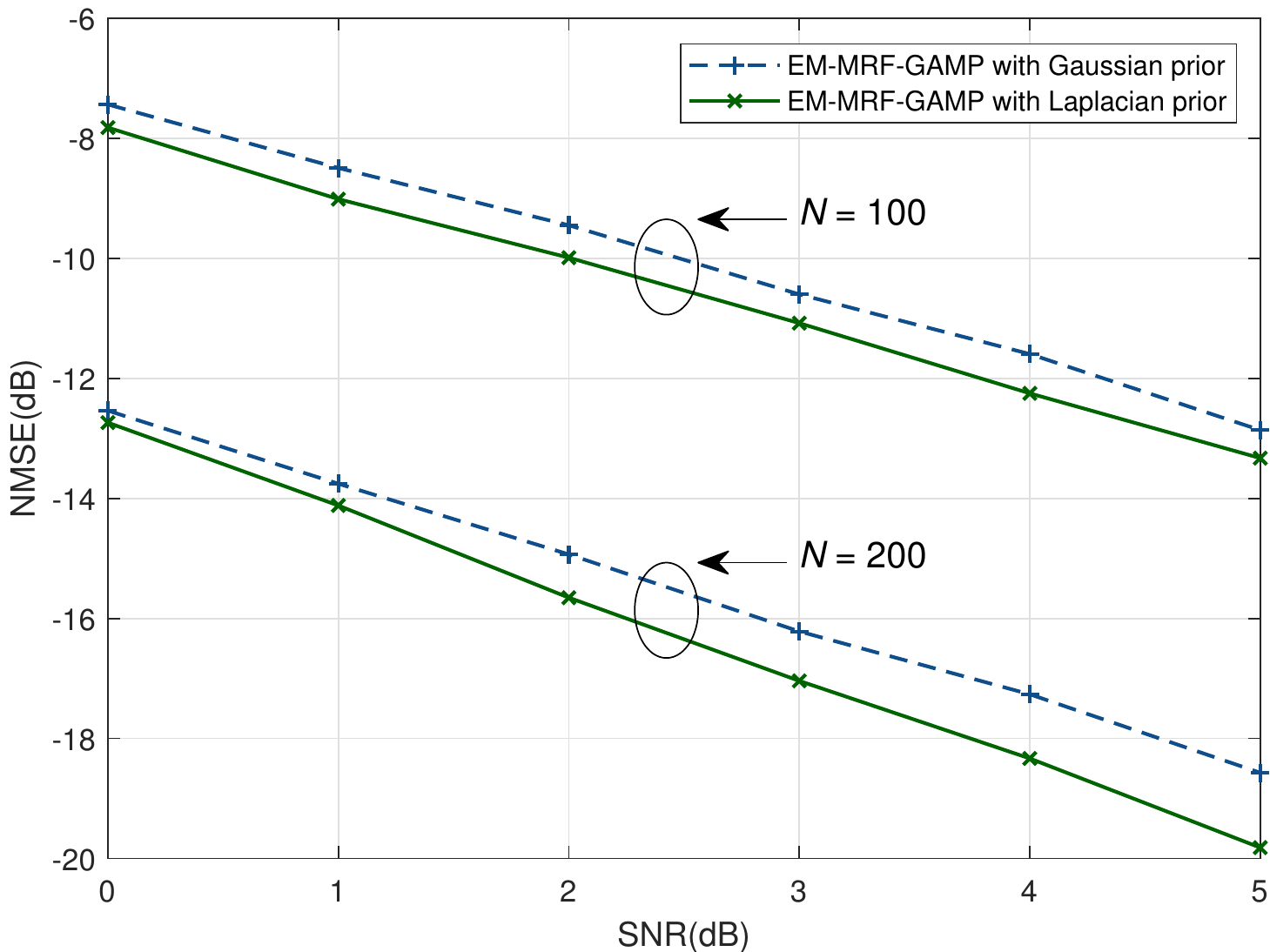}
	\label{NMSE1a}}
	\hspace{0.15in}
	\subfigure[]{\includegraphics[height=5.5cm]{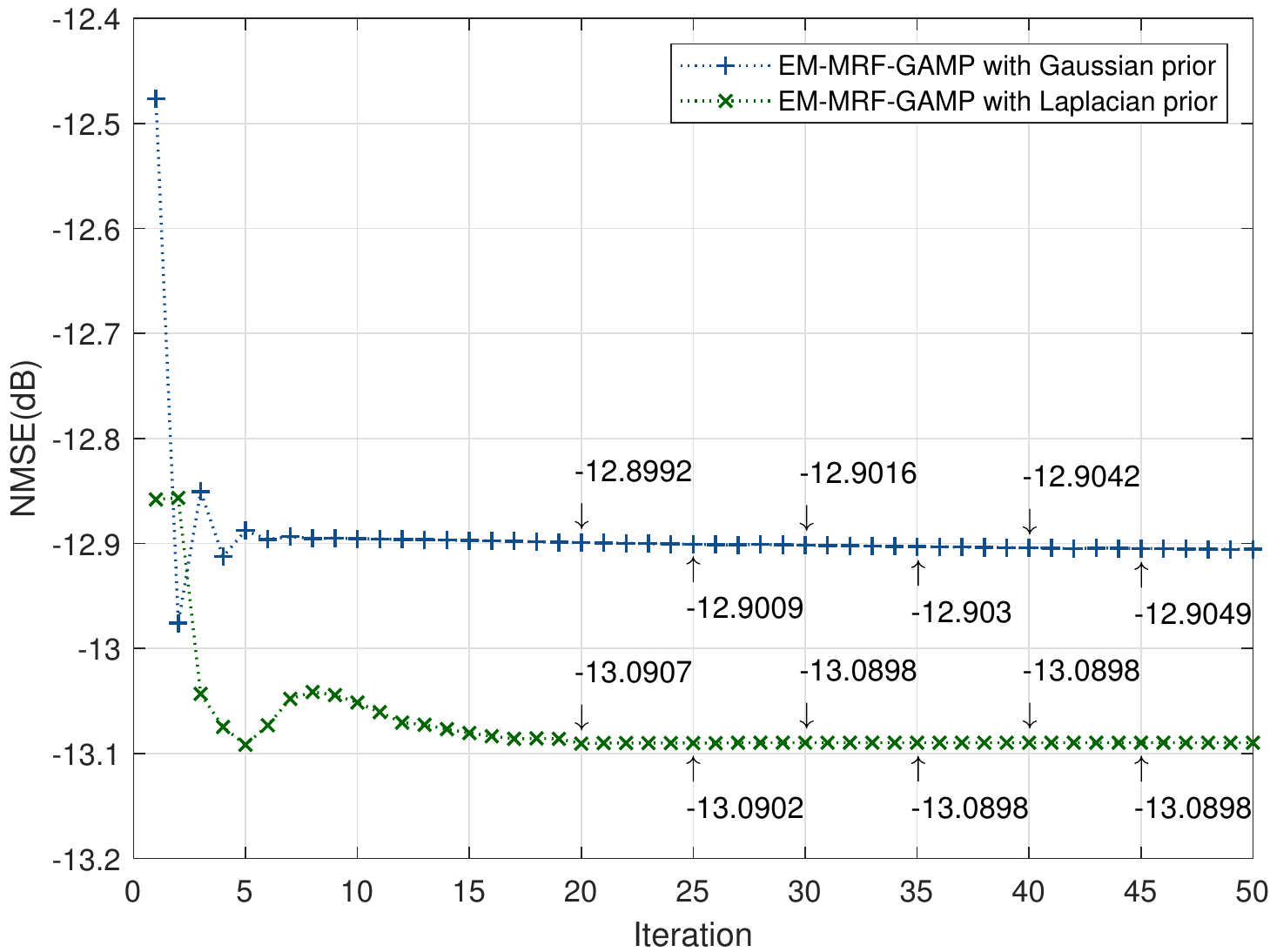}
	\label{NMSE1b}}
	\caption{The NMSEs of MRF-GAMP based algorithms with different priors. $ K_{\mathrm{a}} = 100 $ and $ M = 100 $. a) NMSEs versus the SNR under different numbers of measurements. b) NMSEs versus the iteration number when $ N = 100 $ and $ \text{SNR} = 5 $dB. Precision tolerance is set to be $ \tau = -1 $, i.e., algorithms operate until the maximum number of iterations $ T_{\mathrm{max}} = 50 $ is reached.}
	\label{NMSE1}
\end{figure}

Furthermore, we employ a Bernoulli-Gaussian distributed variable to model the sparse signal $ x $ in EM-MRF-GAMP. In \mbox{Fig. \ref{NMSE1a}}, we see that the EM-MRF-GAMP algorithm with Laplacian prior offers performance gains over EM-MRF-GAMP with Gaussian prior, which is in line with the conclusion drawn in \cite{BSY19} that the Laplacian distribution is more suitable to model the angular domain channel than the Gaussian (mixture) distribution. We also plot in \mbox{Fig. \ref{NMSE1b}} the NMSE performance of both algorithms as a function of iterations. We observe that the algorithm with Laplacian prior converges much faster than the one with Gaussian prior: the former converges after approximately $ 27 $ iterations, while the latter still slightly diverges within $ 50 $ iterations. It is another strength brought by precisely modeling angular domain channel coefficients.

\subsection{Performance of Uncoupled Compressed Sensing Scheme}

Now we examine performance of the proposed UCS scheme with EM-MRF-GAMP as the CS decoder and the slot-balanced $ K $-means assisted clustering decoder. Each $ 96 $-bit user message is divided into fragments of length $ J = 12 $ to send over $ S = 8 $ slots. The total number $ 2^{J} $ of codewords in the common codebook is chosen such that the codeword collision probability is relatively low, meanwhile, the complexity of the EM-MRF-GAMP algorithm is computationally manageable. In URA, the error event probability is defined in the forms of the per-active-user probability of misdetection and the probability of false-alarm, in turn expressed as
\begin{align}
	P_{\text{md}} &= \dfrac{1}{K_{\mathrm{a}}} \sum_{k \in \mathcal{K}_{\mathrm{a}}} p \left( m(k) \notin \mathcal{L} \right), \quad P_{\mathrm{fa}} = \dfrac{\left| \mathcal{L} \backslash \left\lbrace m(k):k \in \mathcal{K}_{\mathrm{a}} \right\rbrace \right| }{\left| \mathcal{L} \right| } \label{Pe}
\end{align}
where $ m(k) $ is a message sequence in the recovered message list $ \mathcal{L} $. The total error rate is counted as the sum of the above error probabilities, i.e., $ P_{\mathrm{e}} = P_{\text{md}} + P_{\text{fa}} $.

\begin{figure}[h]
	\centering
	\subfigure[$ M = 100 $]{\includegraphics[height=5.5cm]{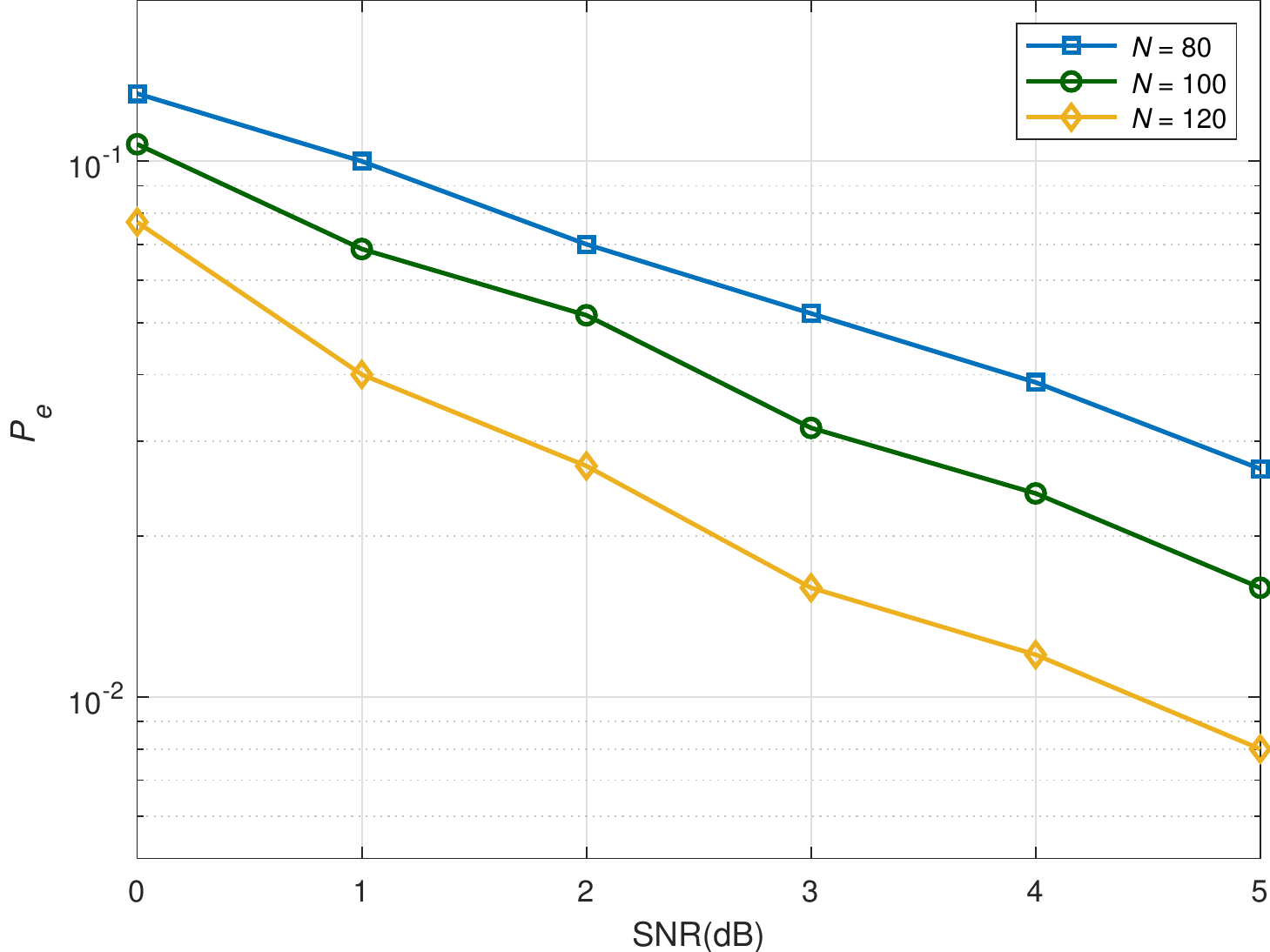}
	\label{simu1a}}
	\hspace{0.15in}
	\subfigure[$ N = 100 $]{\includegraphics[height=5.5cm]{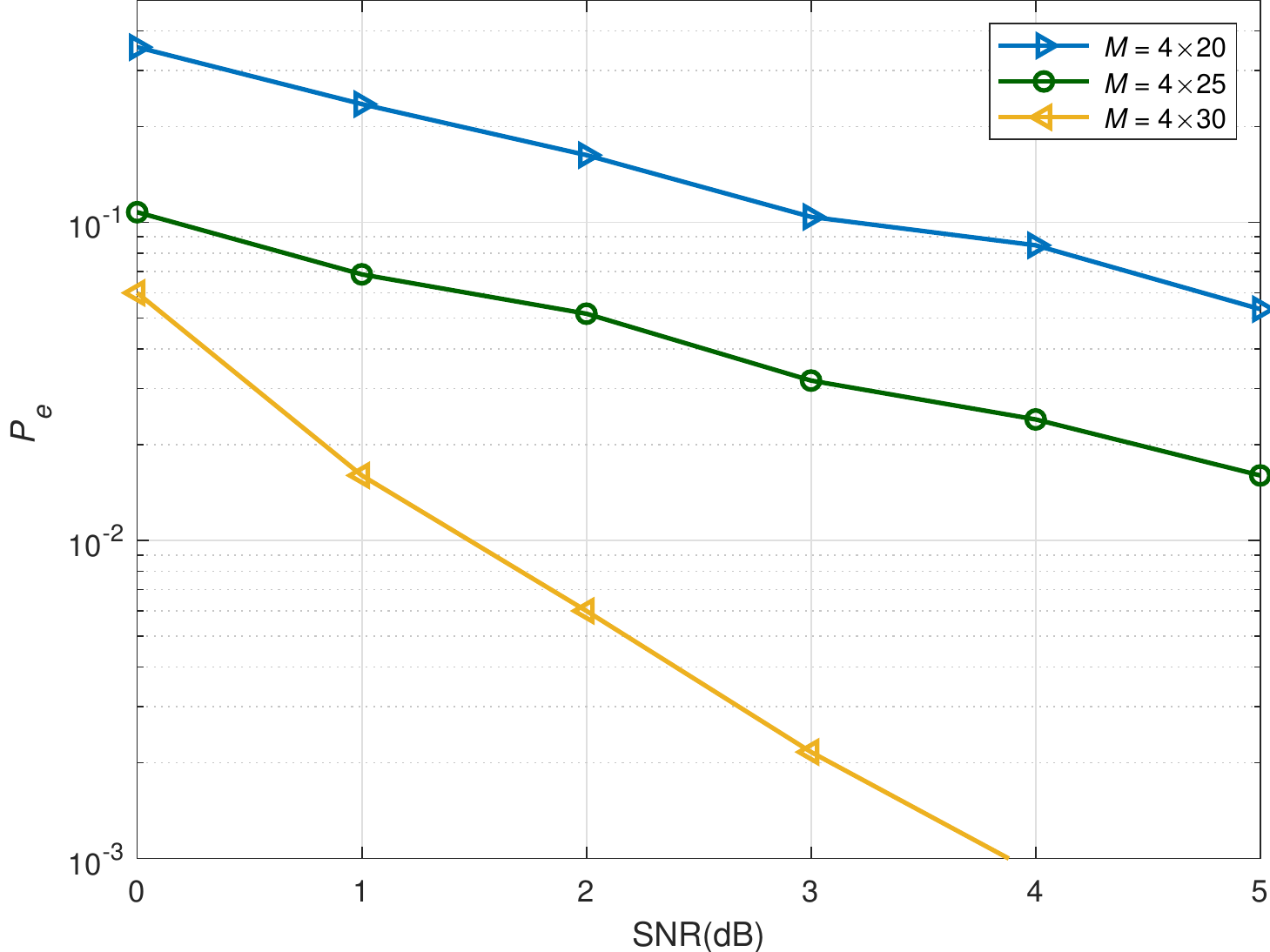}
	\label{simu1b}}
	\caption{Error probabilities as a function of SNR with different values of block-length $ N $ and numbers of antennas $ M $: (a): Fix $ M = 4 \times 25 $ and $ K_{\mathrm{a}} = 100 $, $ N $ varies from $ 80 $ to $ 120 $; (b): Fix $ N = 100 $ and $ K_{\mathrm{a}} = 100 $, $ M $ varies from $ 4 \times 20 $ to $ 4 \times 30 $.}
	\label{simu1}
\end{figure}

\mbox{Fig. \ref{simu1}} demonstrates the error rate of the introduced UCS scheme with different selections of $ M $ and $ N $. It suggests that the system error probability can be decreased by increasing the number of receiving antennas or the coherent block-length. It can be observed in \mbox{Fig. \ref{simu1a}} that if the block-length is shortened by $ 20 $ signal dimensions, one only needs to pay a price of $ 1.0 \sim 1.3 $dB in SNR to achieve the target $ P_{\mathrm{e}} = 0.05 $, as EM-GAMP-MRF is robust to the number of measurements (c.f. \mbox{Fig. \ref{NMSEb}}). \mbox{Fig. \ref{simu1b}} reveals that the error rate improves rapidly when the number of receive antennas is increased. It is owing to the higher resolution offered by more antennas, which provides more dimensional information for measuring channel similarity/difference. Thus, users can be easily distinguished in the angular domain. The total spectral efficiency of the proposed UCS scheme is $ \Psi = \frac{B K_{\mathrm{a}}}{S N} = 12 $ bits per channel use.

The available works of URA in the MIMO scenario \cite{FHJ19, FHJ21, SBM21} all consider i.i.d. MIMO channels. In particular, the CB-CS decoder in \cite{FHJ19, FHJ21} relies highly on the i.i.d. assumption to ensure that the covariance of $ \widetilde{\mathbf{X}} $ is an approximate diagonal matrix. In order to compare with the aforementioned schemes under realistic correlated channels, we put forward the following modified schemes.
\begin{enumerate}[leftmargin=*]
	\item
	\textbf{CCS with CB-CS under correlated channels:} The work of \cite{XWG20} attempts to alleviate the correlation at the transmitting/receiving antenna side to allow the CB-CS recovery method to work under correlated channels. The channel in the transformation domain considered in \cite{XWG20} is approximately independent only when there are rich scatterers between users and the BS. One can refer to \cite{XWG20} for the specific transmission framework design and settings.
	\item
	\textbf{UCS with correlation-aware clustering decoder:} The clustering decoder devised in \cite{SBM21} captures the strong correlation between slot-wise channels of each active user for message stitching. We design a similar correlation-aware clustering decoder in our proposed UCS regime by measuring the distance between the channel vector and the group center based on their correlation, i.e., $ d( k, k' ) = 1 - \frac{\left\langle \mathbf{r}_{k}, \mathbf{c}_{k'} \right\rangle}{\sqrt{\left\langle \mathbf{r}_{k}, \mathbf{r}_{k} \right\rangle \left\langle \mathbf{c}_{k'}, \mathbf{c}_{k'} \right\rangle}} $
	in \eqref{dis}, with $ \left\langle \mathbf{r}, \mathbf{c} \right\rangle = \mathbf{r}^{H} \mathbf{c} $ the Euclidean scalar product. Other system settings are the same as the proposed UCS scheme.
\end{enumerate}
We also provide several intuitive URA schemes for comparison:
\begin{enumerate}[start=3,leftmargin=*]
	\item
	\textbf{CCS with MMV-AMP:} Under the CCS framework, the message is divided into $ 32 $ blocks of size $ J = 12 $ based on the data profile $ \{ 12, 3, 3, \dots, 3, 0, 0, 0 \} $. We apply the MMV-AMP algorithm \cite{LY18} for AD under spatial domain channels, then the tree decoder reconstructs the message list. Such a scheme can be viewed as the MIMO extension of \cite{FJC21}.
	\item
	\textbf{CCS with EM-MRF-GAMP:} Under the CCS framework, we split the message into $ 20 $ blocks of size $ J = 12 $ based on the data profile $ \{ 12, 5, \dots, 5, 4, 0, 0 \} $. Under the angular domain channel, EM-MRF-GAMP acts as the CS decoder for AD.
\end{enumerate}

\begin{figure}[h]
	\centering
	\includegraphics[height=6cm]{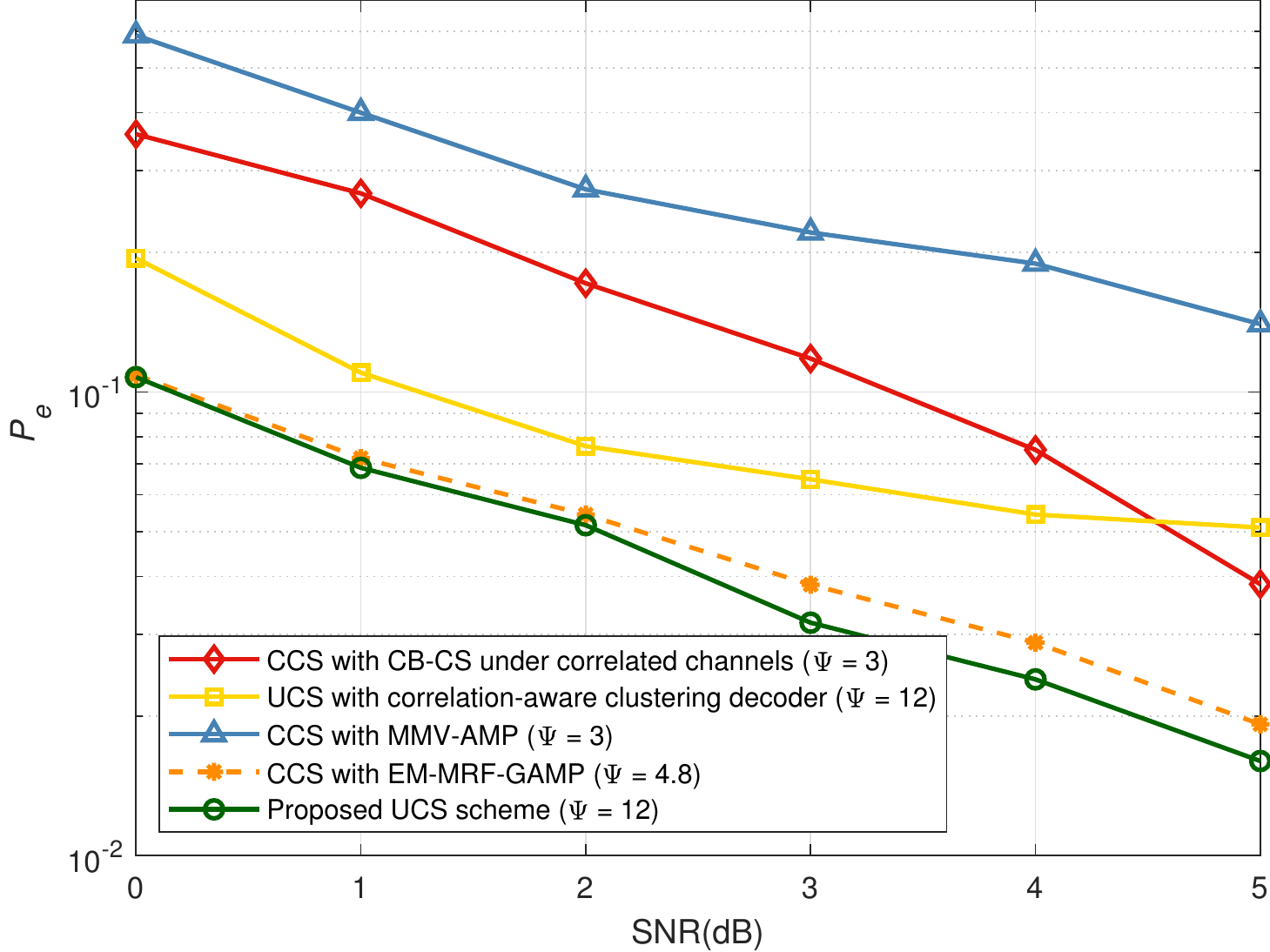}
	\caption{Error probabilities of many schemes as a function of SNR with $ K_{\mathrm{a}}=100 $, $ N = 100 $, and $ M = 100 $.}
	\label{simu2}
\end{figure}

We depict the performance of various URA schemes in \mbox{Fig. \ref{simu2}} as a function of the SNR. Among UCS schemes for URA, the proposed UCS scheme outperforms the one with a correlation-aware clustering decoder. The latter only adopts small-scale fading coefficients for clustering, while we take both large-scale and small-scale fading coefficients into account by the Euclidean distance. It can be seen in \mbox{Fig. \ref{simu2}} that the performance of the CB-CS decoder under correlated channels is not ideal, as the correlation between users is heightened due to the limited number of scatterers. The CCS scheme with MMV-AMP also performs poorly since the MMV-AMP algorithm fails to precisely recover the spatial domain channel with a limited number of measurements, resulting in a high error rate of codeword AD.

It is evident that on the basis of the same AD and CE results offered by the EM-MRF-GAMP based CS decoder, the CCS scheme with a tree-based decoder can ultimately achieve a lower error rate of message stitching than the UCS scheme with a clustering-based decoder by appending many parity check bits. But meanwhile, the corresponding coding rate and spectral efficiency are reduced. We find in \mbox{Fig. \ref{simu2}} that to approach the error rate of the proposed UCS scheme, the CCS regime manifests a spectral efficiency of $ 4.8 $ bits per channel use, which is relatively low compared to that of UCS ($ 12 $ bits per channel use). In general, the proposed uncoupled scheme achieves a low error rate at a high spectral efficiency, which makes it suitable for the massive access scenario.

\begin{figure}[h]
	\centering
	\includegraphics[height=6cm]{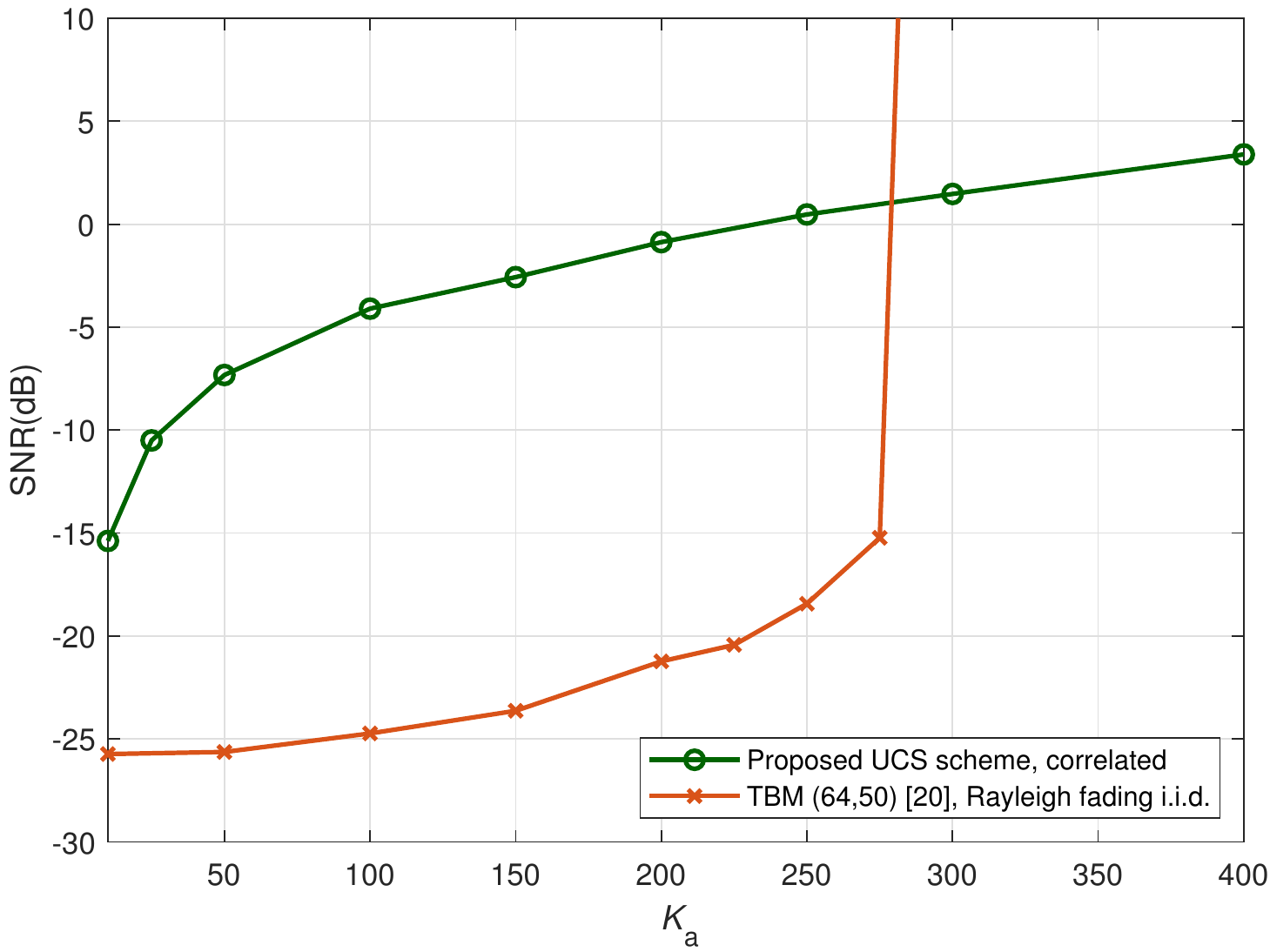}
	\caption{Minimum SNR required to achieve $ P_{\mathrm{md}} \leq 0.1 $ with different values of $ K_{\mathrm{a}} $. $ N_{\mathrm{tot}} = 3200 $ and $ M = 50 $.}
	\label{TBM}
\end{figure}

We also compare the proposed UCS scheme under correlated channels with the tensor-based URA scheme \cite{DLG20} under Rayleigh fading i.i.d. channels. The total block-length for the transmission of $ B = 96 $ bits is $ N_{\mathrm{tot}} = 3200 $. For UCS, messages are sent using $ S = 8 $ slots with $ J = 12 $ and $ N = 400 $. One can refer to \cite{DLG20} for detailed settings of the tensor-based URA scheme with tensor size $ (64,50) $. Focusing on the probability of error $ P_{\mathrm{md}} $ defined in \eqref{Pe}, we depict the SNR required to achieve $ P_{\mathrm{md}} \leq 0.1 $ in \mbox{Fig. \ref{TBM}} with $ M = 50 $. As can be seen from \mbox{Fig. \ref{TBM}}, the tensor-based scheme possesses better energy efficiency, while the proposed UCS scheme supports more potential active users. The tensor-based URA scheme relies on a rank-$ K_{\mathrm{a}} $ tensor decomposition to separate different users' signals. The Kruskal's condition \cite{COV14} for the uniqueness of decomposing a rank-$ K_{\mathrm{a}} $ tensor states that $ K_{\mathrm{a}} $ is positively correlated with the rank of the matrix of active channels. Since correlated channels are of low rank, the tensor-based scheme will support even less active users under correlated channels than under i.i.d. channels.

\section{Conclusion}

URA is a novel paradigm for massive connectivity. We show that by exploiting the rich dimensionality of the sparse angular domain MIMO channel, an uncoupled slotted data transmission can be adopted for URA. We first explore the EM-MRF-GAMP algorithm to retrieve transmitted message sequences and the corresponding channels slot by slot. Afterwards, the similarity of the angular transmission pattern implied in the slot-wise reconstructed channels enables us to design a clustering-based decoder to combine message sequences across slots. We employ the slot-balanced $ K $-means method for message stitching as a constrained assignment problem. Finally, we perform simulation to validate that the presented transmission scheme is reliable with a low error rate in a high spectral efficiency region. 

\begin{appendices}
	
\section{A Note on Angular Domain Transformation}

Considering the $ m_{\mathrm{v}} $-th entry of vector $ \mathbf{u}_{\mathrm{v}} \triangleq \mathbf{U}_{\mathrm{v}}^{H} \mathbf{e}_{\mathrm{v}} \left( \Omega_{k, l}^{\mathrm{v}} \right) $ in \eqref{DFTtrans}, whose magnitude can be calculated as \cite{BSY19}
\begin{align}
	\left| \mathbf{u}_{\mathrm{v}}( m_{\mathrm{v}} ) \right| = \frac{1}{\sqrt{M_{\mathrm{v}}}} \left| \frac{\sin \left( \pi \left[ m_{\mathrm{v}} - 1 - M_{\mathrm{v}} \Omega_{k, l}^{\mathrm{v}} \right] \right)}{\sin \left( \frac{\pi }{M_{\mathrm{v}}} \left[ m_{\mathrm{v}} - 1 - M_{\mathrm{v}} \Omega_{k, l}^{\mathrm{v}} \right] \right)} \right|.
\end{align}
It can be seen that $ | \mathbf{u}_{\mathrm{v}}( m_{\mathrm{v}} ) | $ is maximal for $ \tilde{m}_{\mathrm{v}} $ satisfying
\begin{align}
	\left| \cos ( \phi_{k, l} ) - \dfrac {\tilde{m}_{\mathrm{v}} - 1}{\Delta M_{\mathrm{v}}} \right| < \frac{1}{\Delta M_{\mathrm{v}}}. \label{v1}
\end{align}
Similarly, the magnitude of the $ m_{\mathrm{h}} $-th entry of vector $ \mathbf{u}_{\mathrm{h}} \triangleq \mathbf{U}_{\mathrm{h}}^{H} \mathbf{e}_{\mathrm{h}} \left( \Omega_{k, l}^{\mathrm{h}} \right) $ in \eqref{DFTtrans} is maximal for $ \tilde{m}_{\mathrm{h}} $ satisfying
\begin{align}
	\left| \sin ( \phi_{k, l} ) \cos ( \varphi_{k, l} ) - \dfrac {\tilde{m}_{\mathrm{h}} - 1}{\Delta M_{\mathrm{h}}} \right| < \frac{1}{\Delta M_{\mathrm{h}}}. \label{v2}
\end{align}
Recall \eqref{DFTtrans}, we have that the $ ( m_{\mathrm{v}}, m_{\mathrm{h}} ) $-th element of the angular domain channel $ \mathbf{H} $ has a significant magnitude if there exists a path whose elevation AoA and horizontal AoA verify \eqref{v1} and \eqref{v2} simultaneously.

\section{Calculations of Message Passing Components}

\subsubsection{Message Passing Over $ x_{j m} \to b_{j' m} $}

The message from $ g_{j m} $ to $ b_{j' m} $ can be expressed as
\begin{align}
	\nu_{g_{j m} \to b_{j' m}} &= \dfrac{1}{I_{r}} \int_{x_{j m}} p ( x_{j m} | b_{j' m} ) \mathcal{N} ( x_{j m} ; \widehat{r}_{j m}, \mu_{j m}^{r} ) \notag \\
	&= \varpi_{j m} \delta( b_{j' m} - 1 ) + ( 1 - \varpi_{j m} ) \delta( b_{j' m} + 1 ).
\end{align}
For simplicity, we ignore the subscripts of variables in the following derivations. In the above equation, $ \varpi $ and the normalization constant $ I_{r} $ are respectively given by
\begin{align}
	I_{r} &= \int_{b} \int_{x} p ( x | b ) \mathcal{N} ( x ; \widehat{r}, \mu^{r} ) = \mathcal{N} ( 0 ; \widehat{r}, \mu^{r} ) + \int_{x} \dfrac{\lambda}{2} \exp ( - \lambda | x | ) \mathcal{N} ( x ; \widehat{r}, \mu^{r} )
\end{align}
and
\begin{equation}
	\varpi = \dfrac{1}{I_{r}} \int_{x} \dfrac{\lambda}{2} \exp ( - \lambda | x | ) \mathcal{N} ( x ; \widehat{r}, \mu^{r} ). \label{pii}
\end{equation}
Facing the absolute value within the term $ \psi( x ) \triangleq \frac{\lambda}{2} \exp ( - \lambda | x | ) \mathcal{N} ( x ; \widehat{r}, \mu^{r} ) $, we consider two cases: $ x < 0 $, $ x > 0 $, respectively. For $ x < 0 $, we have
\begin{align}
	\psi( x ) &= \dfrac{\lambda}{2} \exp \left( \lambda x \right) \cdot \dfrac{1}{\sqrt{2 \pi \mu^{r}}} \exp \left( - \dfrac{( x - \widehat{r} )^2}{2 \mu^{r}} \right) = \dfrac{\lambda}{2} \cdot \dfrac{1}{\sqrt{2 \pi \mu^{r}}} \exp \left( - \dfrac{x^2 - 2 \widehat{r} x + \widehat{r}^2 - 2 \lambda \mu^{r} x}{2 \mu^{r}} \right) \notag \\
	&= \dfrac{\lambda}{2} \cdot \dfrac{1}{\sqrt{2 \pi \mu^{r}}} \exp \left( - \dfrac{\left[ x^2 \!-\! ( \widehat{r} \!+\! \lambda \mu^{r} ) \right]^2 \!-\! ( \lambda \mu^{r} )^2 \!-\! 2 \lambda \widehat{r} \mu^{r}}{2 \mu^{r}} \right) \notag \\
	&= \dfrac{\lambda}{2} \exp \left( \dfrac{1}{2} \lambda^2 \mu^{r} + \lambda \widehat{r} \right) \mathcal{N} ( x ; \widehat{r}^{-}, \mu^{r} )
\end{align}
where $ \widehat{r}^{-} = \widehat{r} + \lambda \mu^{r} $. Similarly, for $ x > 0 $, we have
\begin{align}
	\psi( x ) = \dfrac{\lambda}{2} \exp \left( \dfrac{1}{2} \lambda^2 \mu^{r} - \lambda \widehat{r} \right) \mathcal{N} ( x ; \widehat{r}^{+}, \mu^{r} )
\end{align}
where $ \widehat{r}^{+} = \widehat{r} - \lambda \mu^{r} $. The integral of $ \psi( x ) $ on $ x $ is also computed under two conditions as
\begin{align}
	\label{Ip}
	I_{x}^{-} &= \dfrac{\lambda}{2} \exp \left( \dfrac{1}{2} \lambda^2 \mu^{r} + \lambda \widehat{r} \right) \int_{- \infty}^{0} \mathcal{N} ( x ; \widehat{r}^{-}, \mu^{r} ) \mathrm{d} x = \dfrac{\lambda}{2} \exp \left( \dfrac{1}{2} \lambda^2 \mu^{r} + \lambda \widehat{r} \right) \Phi_{\mathcal{N}} \left( \dfrac{-\widehat{r}^{-}}{\sqrt{\mu^{r}}} \right) \\
	\label{Im}
	I_{x}^{+} &= \dfrac{\lambda}{2} \exp \left( \dfrac{1}{2} \lambda^2 \mu^{r} - \lambda \widehat{r} \right) \int_{0}^{\infty} \mathcal{N} ( x ; \widehat{r}^{+}, \mu^{r} ) \mathrm{d} x = \dfrac{\lambda}{2} \exp \left( \dfrac{1}{2} \lambda^2 \mu^{r} - \lambda \widehat{r} \right) \Phi_{\mathcal{N}} \left( \dfrac{\widehat{r}^{+}}{\sqrt{\mu^{r}}} \right)
\end{align}
followed by
\begin{align}
	I_{r} = \mathcal{N} ( 0 ; \widehat{r}, \mu^{r} ) + \left( I_{x}^{-} + I_{x}^{+} \right). \label{Ir}
\end{align}
Plugging \eqref{Ip}, \eqref{Im}, and \eqref{Ir} into \eqref{pii}, we have the closed form of $ \varpi $ expressed in \eqref{pi}.

\subsubsection{Message Updates of Edge/Corner Variable Nodes}

Apart from factor node $ \eta_{j' 1}^{\alpha} $ and the two coupled factor nodes $ g_{j_{\mathrm{re}} 1} $ and $ g_{j_{\mathrm{im}} 1} $, variable node $ b_{j' 1} $ at the corner of the MRF structure receive messages from factor nodes $ \eta_{j', 1, M_{\mathrm{v}} + 1} $ and $ \eta_{j', 1, 2} $ in two directions:
\begin{align}
	\nu_{j' 1}^{d} &= \kappa_{j' 1}^{d} \delta ( b_{j' 1} - 1 ) + \left( 1 - \kappa_{j' 1}^{d} \right) \delta ( b_{j' 1} + 1 )
\end{align}
where $ d \in \{\mathrm{r}, \mathrm{b}\} $, and
\begin{align}
	\kappa_{j' 1}^{\mathrm{r}} &= \tfrac{\varpi_{j_{\mathrm{re}} m_{\mathrm{r}} } \varpi_{j_{\mathrm{im}} m_{\mathrm{r}}} \prod_{k \in \{\mathrm{r,b}\}} \kappa_{j' m_{\mathrm{r}}}^{k} e^{- \alpha_{j'} + \beta_{j'}} + \left( 1 - \varpi_{j_{\mathrm{re}} m_{\mathrm{r}}} \right) \left( 1 - \varpi_{j_{\mathrm{im}} m_{\mathrm{r}}} \right) \prod_{k \in \{\mathrm{r,b}\}} \left( 1 - \kappa_{j' m_{\mathrm{r}}}^{k} \right) e^{\alpha_{j'} - \beta_{j'}}}{\left( e^{\beta_{j'}} + e^{- \beta_{j'}} \right) \left( \varpi_{j_{\mathrm{re}} m_{\mathrm{r}}} \varpi_{j_{\mathrm{im}} m_{\mathrm{r}}} \prod_{k \in \{\mathrm{r,b}\}} \kappa_{j' m_{\mathrm{r}}}^{k} e^{- \alpha_{j'}} + \left( 1 - \varpi_{j_{\mathrm{re}} m_{\mathrm{r}}} \right) \left( 1 - \varpi_{j_{\mathrm{im}} m_{\mathrm{r}}} \right) \prod_{k \in \{\mathrm{r,b}\}} \left( 1 - \kappa_{j' m_{\mathrm{r}}}^{k} \right) e^{\alpha_{j'}} \right)} \\
	\kappa_{j' 1}^{\mathrm{b}} &= \tfrac{\varpi_{j_{\mathrm{re}} m_{\mathrm{b}} } \varpi_{j_{\mathrm{im}} m_{\mathrm{b}}} \prod_{k \in \{\mathrm{r,b}\}} \kappa_{j' m_{\mathrm{b}}}^{k} e^{- \alpha_{j'} + \beta_{j'}} + \left( 1 - \varpi_{j_{\mathrm{re}} m_{\mathrm{b}}} \right) \left( 1 - \varpi_{j_{\mathrm{im}} m_{\mathrm{b}}} \right) \prod_{k \in \{\mathrm{r,b}\}} \left( 1 - \kappa_{j' m_{\mathrm{b}}}^{k} \right) e^{\alpha_{j'} - \beta_{j'}}}{\left( e^{\beta_{j'}} + e^{- \beta_{j'}} \right) \left( \varpi_{j_{\mathrm{re}} m_{\mathrm{b}}} \varpi_{j_{\mathrm{im}} m_{\mathrm{b}}} \prod_{k \in \{\mathrm{r,b}\}} \kappa_{j' m_{\mathrm{b}}}^{k} e^{- \alpha_{j'}} + \left( 1 - \varpi_{j_{\mathrm{re}} m_{\mathrm{b}}} \right) \left( 1 - \varpi_{j_{\mathrm{im}} m_{\mathrm{b}}} \right) \prod_{k \in \{\mathrm{r,b}\}} \left( 1 - \kappa_{j' m_{\mathrm{b}}}^{k} \right) e^{\alpha_{j'}} \right)}
\end{align}
with $ m_{\mathrm{r}} = M_{\mathrm{v}} + 1 $ and $ m_{\mathrm{b}} = 2 $.
The backward message from $ b_{j' 1} $ to $ g_{j 1} $ is represented as
\begin{align}
	\nu_{b_{j' 1} \to g_{j 1}} = \rho_{j 1} \delta( b_{j' 1} - 1 ) + ( 1 - \rho_{j 1} ) \delta( b_{j' 1} + 1 )
\end{align}
with
\begin{align}
	\rho_{j 1} = \dfrac{\varpi_{q 1} \prod_{d \in \{\mathrm{r,b}\}} \kappa_{j' 1}^{d} e^{- \alpha_{j'}}}{\varpi_{q 1} \prod_{d \in \{\mathrm{r,b}\}} \kappa_{j' 1}^{d} e^{- \alpha_{j'}} + ( 1 - \varpi_{q 1} ) \prod_{d \in \{\mathrm{r,b}\}} ( 1- \kappa_{j' 1}^{d} ) e^{\alpha_{j'}}}.
\end{align}

Apart from factor node $ \eta_{j' 2}^{\alpha} $ and the two coupled factor nodes $ g_{j_{\mathrm{re}} 2} $ and $ g_{j_{\mathrm{im}} 2} $, variable node $ b_{j' 2} $ at the edge of the MRF structure receive messages from factor nodes $ \eta_{j', 2, M_{\mathrm{v}} + 2} $, $ \eta_{j', 1, 2} $ and $ \eta_{j', 2, 3} $ in three directions:
\begin{align}
	\nu_{j' 2}^{d} &= \kappa_{j' 2}^{d} \delta ( b_{j' 2} - 1 ) + \left( 1 - \kappa_{j' 2}^{d} \right) \delta ( b_{j' 2} + 1 )
\end{align}
where $ d \in \{\mathrm{t}, \mathrm{r}, \mathrm{b}\} $, and
\begin{align}
	\kappa_{j' 2}^{\mathrm{t}} &= \tfrac{\varpi_{j_{\mathrm{re}} m_{\mathrm{t}} } 	\varpi_{j_{\mathrm{im}} m_{\mathrm{t}}} \kappa_{j' m_{\mathrm{t}}}^{\mathrm{r}} e^{- \alpha_{j'} + \beta_{j'}} + \left( 1 - \varpi_{j_{\mathrm{re}} m_{\mathrm{t}}} \right) \left( 1 - \varpi_{j_{\mathrm{im}} m_{\mathrm{t}}} \right) \left( 1 - \kappa_{j' m_{\mathrm{t}}}^{\mathrm{r}} \right) e^{\alpha_{j'} - \beta_{j'}}}{\left( e^{\beta_{j'}} + e^{- \beta_{j'}} \right) \left( \varpi_{j_{\mathrm{re}} m_{\mathrm{t}}} \varpi_{j_{\mathrm{im}} m_{\mathrm{t}}} \kappa_{j' m_{\mathrm{t}}}^{\mathrm{r}} e^{- \alpha_{j'}} + \left( 1 - \varpi_{j_{\mathrm{re}} m_{\mathrm{t}}} \right) \left( 1 - \varpi_{j_{\mathrm{im}} m_{\mathrm{t}}} \right) \left( 1 - \kappa_{j' m_{\mathrm{t}}}^{\mathrm{r}} \right) e^{\alpha_{j'}} \right)} \\
	\kappa_{j' 2}^{\mathrm{r}} &= \tfrac{\varpi_{j_{\mathrm{re}} m_{\mathrm{r}} } 	\varpi_{j_{\mathrm{im}} m_{\mathrm{r}}} \prod_{k \in \{\mathrm{r,t,b}\}} \kappa_{j' m_{\mathrm{r}}}^{k} e^{- \alpha_{j'} + \beta_{j'}} + \left( 1 - \varpi_{j_{\mathrm{re}} m_{\mathrm{r}}} \right) \left( 1 - \varpi_{j_{\mathrm{im}} m_{\mathrm{r}}} \right) \prod_{k \in \{\mathrm{r,t,b}\}} \left( 1 - \kappa_{j' m_{\mathrm{r}}}^{k} \right) e^{\alpha_{j'} - \beta_{j'}}}{\left( e^{\beta_{j'}} + e^{- \beta_{j'}} \right) \left( \varpi_{j_{\mathrm{re}} m_{\mathrm{r}}} \varpi_{j_{\mathrm{im}} m_{\mathrm{r}}} \prod_{k \in \{\mathrm{r,t,b}\}} \kappa_{j' m_{\mathrm{r}}}^{k} e^{- \alpha_{j'}} + \left( 1 - \varpi_{j_{\mathrm{re}} m_{\mathrm{r}}} \right) \left( 1 - \varpi_{j_{\mathrm{im}} m_{\mathrm{r}}} \right) \prod_{k \in \{\mathrm{r,t,b}\}} \left( 1 - \kappa_{j' m_{\mathrm{r}}}^{k} \right) e^{\alpha_{j'}} \right)} \\
	\kappa_{j' 2}^{\mathrm{b}} &= \tfrac{\varpi_{j_{\mathrm{re}} m_{\mathrm{b}} } 	\varpi_{j_{\mathrm{im}} m_{\mathrm{b}}} \prod_{k \in \{\mathrm{r,b}\}} \kappa_{j' m_{\mathrm{b}}}^{k} e^{- \alpha_{j'} + \beta_{j'}} + \left( 1 - \varpi_{j_{\mathrm{re}} m_{\mathrm{b}}} \right) \left( 1 - \varpi_{j_{\mathrm{im}} m_{\mathrm{b}}} \right) \prod_{k \in \{\mathrm{r,b}\}} \left( 1 - \kappa_{j' m_{\mathrm{b}}}^{k} \right) e^{\alpha_{j'} - \beta_{j'}}}{\left( e^{\beta_{j'}} + e^{- \beta_{j'}} \right) \left( \varpi_{j_{\mathrm{re}} m_{\mathrm{b}}} \varpi_{j_{\mathrm{im}} m_{\mathrm{b}}} \prod_{k \in \{\mathrm{r,b}\}} \kappa_{j' m_{\mathrm{b}}}^{k} e^{- \alpha_{j'}} + \left( 1 - \varpi_{j_{\mathrm{re}} m_{\mathrm{b}}} \right) \left( 1 - \varpi_{j_{\mathrm{im}} m_{\mathrm{b}}} \right) \prod_{k \in \{\mathrm{r,b}\}} \left( 1 - \kappa_{j' m_{\mathrm{b}}}^{k} \right) e^{\alpha_{j'}} \right)}
\end{align}
with $ m_{\mathrm{r}} = M_{\mathrm{v}} + 2 $, $ m_{\mathrm{t}} = 1 $, and $ m_{\mathrm{b}} = 3 $.
The backward message from $ b_{j' 2} $ to $ g_{j 2} $ is represented as
\begin{align}
	\nu_{b_{j' 2} \to g_{j 2}} = \rho_{j 2} \delta( b_{j' 2} - 1 ) + ( 1 - \rho_{j 2} ) \delta( b_{j' 2} + 1 )
\end{align}
with
\begin{align}
	\rho_{j 2} = \dfrac{\varpi_{q 2} \prod_{d \in \{\mathrm{r,t,b}\}} \kappa_{j' 2}^{d} e^{- 	\alpha_{j'}}}{\varpi_{q 2} \prod_{d \in \{\mathrm{r,t,b}\}} \kappa_{j' 2}^{d} e^{- \alpha_{j'}} + ( 1 - \varpi_{q 2} ) \prod_{d \in \{\mathrm{r,t,b}\}} ( 1- \kappa_{j' 2}^{d} ) e^{\alpha_{j'}}}.
\end{align}

\subsubsection{Derivations of $ \widehat{x}_{j m} $ and $ \mu_{j m}^{x} $}

Consider the marginal posterior \eqref{marg}, the integral items $ \int_{x} x \mathcal{N} ( x ; \widehat{r}, \mu^{r} ) \nu_{g \to x} $ and $ \int_{x} x^{2} \mathcal{N} ( x ; \widehat{r}, \mu^{r} ) \nu_{g \to x} $ can be calculated as
\begin{align}
	\label{57}
	&\int_{x} x \mathcal{N} ( x ; \widehat{r}, \mu^{r} ) \nu_{g \to x} \notag \\ &= \dfrac{\lambda}{2} \exp \left( \dfrac{1}{2} \lambda^2 \mu^{r} + \lambda \widehat{r} \right) \int_{- \infty}^{0} x \mathcal{N} ( x ; \widehat{r}^{-}, \mu^{r} ) \mathrm{d} x  + \dfrac{\lambda}{2} \exp \left( \dfrac{1}{2} \lambda^2 \mu^{r} - \lambda \widehat{r} \right) \int_{0}^{\infty} x \mathcal{N} ( x ; \widehat{r}^{+}, \mu^{r} ) \mathrm{d} x \notag \\
	&= \tfrac{I_{x}^{-}}{\Phi_{\mathcal{N}} \left( \tfrac{-\widehat{r}^{-}}{\sqrt{\mu^{r}}} \right)} \cdot \dfrac{1}{\sqrt{2 \pi}} \int_{- \infty}^{\frac{-\widehat{r}^{-}}{\sqrt{\mu^{r}}}} ( \sqrt{\mu^{r}} t + \widehat{r}^{-} ) e^{- \frac{t^2}{2}} \mathrm{d} t
	+ \tfrac{I_{x}^{+}}{\Phi_{\mathcal{N}} \left( \tfrac{\widehat{r}^{+}}{\sqrt{\mu^{r}}} \right)} \cdot \dfrac{1}{\sqrt{2 \pi}} \int_{- \infty}^{\frac{\widehat{r}^{+}}{\sqrt{\mu^{r}}}} ( - \sqrt{\mu^{r}} t + \widehat{r}^{+} ) e^{- \frac{t^2}{2}} \mathrm{d} t \notag \\
	&= \rho I_{x}^{-} \left[ \widehat{r}^{-} - \mu^{r} \frac{\mathcal{N} ( 0 ; \widehat{r}^{-}, \mu^{r} )}{\Phi_{\mathcal{N}} \left( -\widehat{r}^{-} / \sqrt{\mu^{r}} \right)} \right] + \rho I_{x}^{+} \left[ \widehat{r}^{+} + \mu^{r} \frac{\mathcal{N} ( 0 ; \widehat{r}^{+}, \mu^{r} )}{\Phi_{\mathcal{N}} \left( \widehat{r}^{+} / \sqrt{\mu^{r}} \right)} \right] \\
	\label{58}
	&\int_{x} x^{2} \mathcal{N} ( x ; \widehat{r}, \mu^{r} ) \nu_{g \to x} \notag \\
	&= \dfrac{\rho \lambda}{2} \exp \left( \dfrac{1}{2} \lambda^2 \mu^{r} + \lambda \widehat{r} \right) \int_{- \infty}^{0} x^{2} \mathcal{N} ( x ; \widehat{r}^{-}, \mu^{r} ) \mathrm{d} x + \dfrac{\rho \lambda}{2} \exp \left( \dfrac{1}{2} \lambda^2 \mu^{r} - \lambda \widehat{r} \right) \int_{0}^{\infty} x^{2} \mathcal{N} ( x ; \widehat{r}^{+}, \mu^{r} ) \mathrm{d} x \notag \\
	&= \tfrac{\rho I_{x}^{-}}{\Phi_{\mathcal{N}} \left( \tfrac{-\widehat{r}^{-}}{\sqrt{\mu^{r}}} \right)} \cdot \dfrac{1}{\sqrt{2 \pi}} \int_{- \infty}^{\frac{-\widehat{r}^{-}}{\sqrt{\mu^{r}}}} ( \sqrt{\mu^{r}} t + \widehat{r}^{-} )^2 e^{- \frac{t^2}{2}} \mathrm{d} t + \tfrac{\rho I_{x}^{+}}{\Phi_{\mathcal{N}} \left( \tfrac{\widehat{r}^{+}}{\sqrt{\mu^{r}}} \right)} \cdot \dfrac{1}{\sqrt{2 \pi}} \int_{- \infty}^{\frac{\widehat{r}^{+}}{\sqrt{\mu^{r}}}} ( - \sqrt{\mu^{r}} t + \widehat{r}^{+} )^2 e^{- \frac{t^2}{2}} \mathrm{d} t \notag \\
	&= \rho I_{x}^{-} \left[ ( \widehat{r}^{-} )^{2} + \mu^{r} - \frac{\widehat{r}^{-} \mu^{r} \mathcal{N} ( 0 ; \widehat{r}^{-}, \mu^{r} )}{\Phi_{\mathcal{N}} \left( -\widehat{r}^{-} / \sqrt{\mu^{r}} \right)} \right] + \rho I_{x}^{+} \left[ ( \widehat{r}^{+} )^{2} + \mu^{r} + \frac{\widehat{r}^{+} \mu^{r} \mathcal{N} ( 0 ; \widehat{r}^{+}, \mu^{r} )}{\Phi_{\mathcal{N}} \left( \widehat{r}^{+} / \sqrt{\mu^{r}} \right)} \right].
\end{align}
Combining the results of \eqref{57}, \eqref{58}, and the normalization constant $ I_x $ \eqref{Ix}, the mean and variance of $ p( x_{j m} | \mathbf{Y} ) $ can be easily achieved.

\section{Proof of Corollary 1}

PUPE is equivalent to the probability $ p(\left\| \widehat{\boldsymbol{x}} \right\|^{2} > \upsilon | \boldsymbol{x} = \mathbf{0}) $ \cite{CSY18}, i.e,
\begin{align}
	\mathrm{PUPE} = \int_{\left\| \widehat{\boldsymbol{x}} \right\|^{2} > \upsilon} p(\widehat{\boldsymbol{x}} | \boldsymbol{x} = \mathbf{0}) d\widehat{\boldsymbol{x}} = \int_{\left\| \widehat{\boldsymbol{x}} \right\|^{2} > \upsilon} \dfrac{\exp \left( - \left\| \widehat{\boldsymbol{x}} \right\|^{2} \varrho^{-2} \right)}{\pi^{M} \varrho^{2M}}  d\widehat{\boldsymbol{x}} \overset{(a)}{=} \dfrac{\overline{\Gamma}(M, \upsilon \varrho^{-2})}{\Gamma(M)}
\end{align}
where $ (a) $ is obtained by treating the integral of $ \widehat{\boldsymbol{x}} $ as the cumulative distribution function of a $ \chi^{2} $ distribution with $ 2M $ degrees of freedom.

Now with $ \upsilon = c M \varrho^{2} $, we have \cite{G11}
\begin{align}
	\lim\limits_{M \to \infty} \dfrac{\overline{\Gamma}(M, \upsilon \varrho^{-2})}{\Gamma(M)} = \dfrac{1}{2} \operatorname{erfc} \left( C \sqrt{\tfrac{1}{2} M} \right) + \dfrac{\exp (-\frac{1}{2} M C^{2})}{\sqrt{2 \pi M}} \sum_{i=0}^{\infty} \dfrac{\mathcal{C}_{i}(M)}{M^{i}}
\end{align}
where $ C = \sqrt{2 (c - 1 - \log c)} > 0 $ for $ c > 1 $, and $ \mathcal{C}_{0}(M) = 1 $, $ \mathcal{C}_{1}(M) = 0 $, $ \mathcal{C}_{i}(M) = \mathcal{C}_{i-2}(M) + L_{i}^{(1-M-i)}(1-M) $ with $ L_{i}^{M} $ the Laguerre polynomials. It is known that the complementary error function
\begin{align}
	\operatorname{erfc}(x) = \dfrac{\exp (-x^{2})}{\sqrt{\pi} x} \left( 1 + o(x^{-2}) \right).
\end{align}
Therefore, we have
\begin{align}
	\lim\limits_{M \to \infty} \dfrac{\overline{\Gamma}(M, \upsilon \varrho^{-2})}{\Gamma(M)} = \lim\limits_{M \to \infty} \left[ \dfrac{\exp (-\frac{1}{2} M C^{2} )}{C \sqrt{\pi M / 2}} \left( 1 + o \left( \dfrac{1}{M} \right) \right) + o\left( \dfrac{\exp (-M)}{\sqrt{M}} \right) \right] = 0.
\end{align}
Corollary 1 is thus proved.

\end{appendices}

\bibliographystyle{IEEEtran}
\bibliography{reference}

\begin{thebibliography}{10}
\providecommand{\url}[1]{#1}
\csname url@samestyle\endcsname
\providecommand{\newblock}{\relax}
\providecommand{\bibinfo}[2]{#2}
\providecommand{\BIBentrySTDinterwordspacing}{\spaceskip=0pt\relax}
\providecommand{\BIBentryALTinterwordstretchfactor}{4}
\providecommand{\BIBentryALTinterwordspacing}{\spaceskip=\fontdimen2\font plus
\BIBentryALTinterwordstretchfactor\fontdimen3\font minus
  \fontdimen4\font\relax}
\providecommand{\BIBforeignlanguage}[2]{{%
\expandafter\ifx\csname l@#1\endcsname\relax
\typeout{** WARNING: IEEEtran.bst: No hyphenation pattern has been}%
\typeout{** loaded for the language `#1'. Using the pattern for}%
\typeout{** the default language instead.}%
\else
\language=\csname l@#1\endcsname
\fi
#2}}
\providecommand{\BIBdecl}{\relax}
\BIBdecl

\bibitem{DSG17}
Z.~Dawy, W.~Saad, A.~Ghosh, J.~G. Andrews, and E.~Yaacoub, ``Toward massive
  machine type cellular communications,'' \emph{{IEEE} Wireless Commun.},
  vol.~24, no.~1, pp. 120--128, Feb. 2017.

\bibitem{WGZ20}
Y.~Wu, X.~Gao, S.~Zhou, W.~Yang, Y.~Polyanskiy, and G.~Caire, ``Massive access
  for future wireless communication systems,'' \emph{{IEEE} Wireless Commun.},
  vol.~27, no.~4, pp. 148--156, Aug. 2020.

\bibitem{CNY21}
X.~Chen, D.~W.~K. Ng, W.~Yu, E.~G. Larsson, N.~Al-Dhahir, and R.~Schober,
  ``Massive access for {5G} and beyond,'' \emph{{IEEE} J. Sel. Areas in
  Commun.}, vol.~39, no.~3, pp. 615--637, Mar. 2021.

\bibitem{HHN13}
M.~Hasan, E.~Hossain, and D.~Niyato, ``Random access for machine-to-machine
  communication in {LTE}-advanced networks: Issues and approaches,''
  \emph{{IEEE} Commun. Mag.}, vol.~51, no.~6, pp. 86--93, Jun. 2013.

\bibitem{SL18}
K.~Senel and E.~G. Larsson, ``Grant-free massive {MTC}-enabled massive {MIMO}:
  A comprehensive sensing approach,'' \emph{{IEEE} Trans. Commun.}, vol.~66,
  no.~12, pp. 6164--6175, Dec. 2018.

\bibitem{LY18}
L.~Liu and W.~Yu, ``Massive connectivity with massive {MIMO}---{Part} {I}:
  Device activity detection and channel estimation,'' \emph{{IEEE} Trans.
  Signal Process.}, vol.~66, no.~11, pp. 2933--2946, Jun. 2018.

\bibitem{LWS20}
Y.~Li, W.~Wang, X.~Song, X.~Gao, L.~Wang, and G.~P. Fettweis, ``Unified
  iterative receiver design in uplink grant-free massive {MIMO} {SCMA}
  systems,'' in \emph{Proc. IEEE Glob. Commun. Conf. (GLOBECOM)}, Dec. 2020,
  pp. 1--6.

\bibitem{P07}
Y.~Polyanskiy, ``A perspective on massive random-access,'' in \emph{Proc.
  {IEEE} Int. Symp. Inf. Theor. (ISIT)}, Jun. 2017, pp. 2523--2527.

\bibitem{CSY18}
Z.~Chen, F.~Sohrabi, and W.~Yu, ``Sparse activity detection for massive
  connectivity,'' \emph{{IEEE} Trans. Signal Process.}, vol.~66, no.~7, pp.
  1890--1904, Apr. 2018.

\bibitem{CT20}
R.~Calderbank and A.~Thompson, ``{CHIRRUP}: {a} practical algorithm for
  unsourced multiple access,'' \emph{Information and Inference}, vol.~9, pp.
  875--897, Dec. 2020.

\bibitem{RO21}
E.~Romanov and O.~Ordentlich, ``On compressed sensing of binary signals for the
  unsourced random access channel,'' \emph{Entropy}, vol.~23, no.~5, p. 605,
  May 2021.

\bibitem{FJC21}
A.~Fengler, P.~Jung, and G.~Caire, ``{SPARCs} for unsourced random access,''
  \emph{{IEEE} Trans. Inf. Theory}, vol.~67, no.~10, pp. 6894--6915, Oct. 2021.

\bibitem{ACN20'1}
V.~K. Amalladinne, J.~F. Chamberland, and K.~R. Narayanan, ``A coded compressed
  sensing scheme for unsourced multiple access,'' \emph{{IEEE} Trans. Inf.
  Theory}, vol.~66, no.~10, pp. 6509--6533, Oct. 2020.

\bibitem{ACN20'2}
------, ``An enhanced decoding algorithm for coded compressed sensing,'' in
  \emph{Proc. {IEEE} Int. Conf. Acoust. Speech Signal Process. (ICASSP)}, May
  2020, pp. 5270--5274.

\bibitem{ADP20}
V.~K. Amalladinne, A.~Department, K.~Pradhan, C.~Rush, J.~F. Chamberland, and
  K.~R. Narayanan, ``On approximate message passing for unsourced access with
  coded compressed sensing,'' in \emph{Proc. {IEEE} Int. Symp. Inf. Theor.
  (ISIT)}, Jun. 2020, pp. 2995--3000.

\bibitem{FHJ19}
A.~Fengler, S.~Haghighatshoar, P.~Jung, and G.~Caire, ``Grant-free massive
  random access with a massive {MIMO} receiver,'' in \emph{Proc. 53rd Asilomar
  Conf. Signals Syst. Comput. (ACSSC)}, Nov. 2019, pp. 23--30.

\bibitem{FHJ21}
------, ``Non-{Bayesian} activity detection, large-scale fading coefficient
  estimation, and unsourced random access with a massive {MIMO} receiver,''
  \emph{{IEEE} Trans. Inf. Theory}, vol.~67, no.~5, pp. 2925--2951, May 2021.

\bibitem{HJC18}
S.~Haghighatshoar, P.~Jung, and G.~Caire, ``Improved scaling law for activity
  detection in massive {MIMO} systems,'' in \emph{Proc. {IEEE} Int. Symp. Inf.
  Theor. (ISIT)}, Jun. 2018, pp. 381--385.

\bibitem{FJC21'2}
A.~Fengler, P.~Jung, and G.~Caire, ``Pilot-based unsourced random access with a
  massive mimo receiver in the quasi-static fading regime,'' in \emph{Proc.
  IEEE Workshop Signal Process. Adv. Wireless Commun. (SPAWC)}, 2021, pp.
  356--360.

\bibitem{DLG20}
A.~Decurninge, I.~Land, and M.~Guillaud, ``Tensor-based modulation for
  unsourced massive random access,'' \emph{{IEEE} Wireless Commun. Lett.},
  vol.~10, no.~3, pp. 552--556, Nov. 2020.

\bibitem{DLG21}
------, ``Tensor decomposition bounds for {TBM}-based massive access,'' in
  \emph{Proc. {IEEE} Workshop Signal Process. Adv. Wireless Commun. (SPAWC)},
  Sep. 2021, pp. 346--350.

\bibitem{SBM21}
V.~Shyianov, F.~Bellili, A.~Mezghani, and E.~Hossain, ``Massive unsourced
  random access based on uncoupled compressive sensing: Another blessing of
  massive {MIMO},'' \emph{{IEEE} J. Sel. Areas in Commun.}, vol.~39, no.~3, pp.
  820--834, Mar. 2021.

\bibitem{GDW15}
Z.~Gao, L.~Dai, Z.~Wang, and S.~Chen, ``Spatially common sparsity based
  adaptive channel estimation and feedback for {FDD} massive {MIMO},''
  \emph{{IEEE} Trans. Signal Process.}, vol.~63, no.~23, pp. 6169--6183, Dec.
  2015.

\bibitem{KGW20}
M.~Ke, Z.~Gao, Y.~Wu, X.~Gao, and R.~Schober, ``Compressive sensing based
  adaptive active user detection and channel estimation: Massive access meets
  massive {MIMO},'' \emph{{IEEE} Trans. Signal Process.}, vol.~68, pp.
  764--779, Jan. 2020.

\bibitem{YGX15}
L.~You, X.~Gao, X.~Xia, N.~Ma, and Y.~Peng, ``Pilot reuse for massive {MIMO}
  transmission over spatially correlated {Rayleigh} fading channels,''
  \emph{{IEEE} Trans. Wireless Commun.}, vol.~14, no.~6, pp. 3352--3366, Jun.
  2015.

\bibitem{R11}
S.~Rangan, ``Generalized approximate message passing for estimation with random
  linear mixing,'' in \emph{Proc. {IEEE} Int. Symp. Inf. Theor. (ISIT)}, Jun.
  2011, pp. 2168--2172.

\bibitem{SS11}
S.~Som and P.~Schniter, ``Approximate message passing for recovery of sparse
  signals with {Markov}-random-field support structure,'' in \emph{Proc. Int.
  Conf. Mach. Learn. (ICML)}, 2011, pp. 1--15.

\bibitem{BSY19}
F.~Bellili, F.~Sohrabi, and W.~Yu, ``Generalized approximate message passing
  for massive {MIMO} mmwave channel estimation with {Laplacian} prior,''
  \emph{{IEEE} Trans. Commun.}, vol.~67, no.~5, pp. 3205--3219, May 2019.

\bibitem{S02}
A.~M. Sayeed, ``Deconstructing multiantenna fading channels,'' \emph{{IEEE}
  Trans. Signal Process.}, vol.~50, no.~10, pp. 2563--2579, Oct. 2002.

\bibitem{VS13}
J.~P. Vila and P.~Schniter, ``Expectation-maximization {Gaussian}-mixture
  approximate message passing,'' \emph{{IEEE} Trans. Signal Process.}, vol.~61,
  no.~19, pp. 4658--4672, Oct. 2013.

\bibitem{ZYH20}
M.~Zhang, X.~Yuan, and Z.-Q. He, ``Variance state propagation for structured
  sparse bayesian learning,'' \emph{{IEEE} Trans. Signal Process.}, vol.~68,
  pp. 2386--2400, Mar. 2020.

\bibitem{KFL01}
F.~R. Kschischang, B.~J. Frey, and H.~. Loeliger, ``Factor graphs and the
  sum-product algorithm,'' \emph{{IEEE} Trans. Inf. Theory}, vol.~47, no.~2,
  pp. 498--519, Feb. 2001.

\bibitem{DMM09}
D.~L. Donoho, A.~Maleki, and A.~Montanari, ``Message-passing algorithms for
  compressed sensing,'' \emph{Proc. Nat. Acad. Sci. USA}, vol. 106, no.~45, pp.
  18\,914--18\,919, Nov. 2009.

\bibitem{BM11}
M.~Bayati and A.~Montanari, ``The dynamics of message passing on dense graphs,
  with applications to compressed sensing,'' \emph{{IEEE} Trans. Inf. Theory},
  vol.~57, no.~2, pp. 764--785, Feb. 2011.

\bibitem{BMN20}
R.~Berthier, A.~Montanari, and P.-M. Nguyen, ``State evolution for approximate
  message passing with non-separable functions,'' \emph{Information and
  Inference}, vol.~9, pp. 33--79, 2020.

\bibitem{MRB17}
\BIBentryALTinterwordspacing
Y.~Ma, C.~Rush, and D.~Baron, ``Analysis of approximate message passing with a
  class of non-separable denoisers,'' 2017. [Online]. Available:
  \url{https://arxiv.org/abs/1705.03126}
\BIBentrySTDinterwordspacing

\bibitem{HW79}
J.~A. Hartigan and M.~A. Wong, ``Algorithm {AS} 136: A k-means clustering
  algorithm,'' \emph{Applied Statistics}, vol.~28, no.~1, pp. 100--108, 1979.

\bibitem{K55}
H.~W. Kuhn, ``The {Hungarian} method for the assignment problem,'' \emph{Nav.
  Res. Logistics Quart.}, vol.~2, no.~1, pp. 83--97, Mar. 1955.

\bibitem{WCR01}
K.~Wagstaff, C.~Cardie, S.~Rogers, and S.~Schroedl, ``Constrained k-means
  clustering with background knowledge,'' in \emph{Proc. Int. Conf. Mach.
  Learn. (ICML)}, vol.~1, 2001, pp. 577--584.

\bibitem{MF14}
M.~Malinen and P.~Fr{\"{a}}nti, ``Balanced k-means for clustering,'' in
  \emph{Proc. Joint IAPR Int. Workshop Struct. Syntactic Statist. Pattern
  Recognit. (S+SSPR)}, 2014, pp. 32--41.

\bibitem{WWA18}
S.~Wu, C.~Wang, e.~M.~Aggoune, M.~M. Alwakeel, and X.~You, ``A general {3-D}
  non-stationary {5G} wireless channel model,'' \emph{{IEEE} Trans. Commun.},
  vol.~66, no.~7, pp. 3065--3078, Jul. 2018.

\bibitem{FLE20}
J.~Flordelis, X.~Li, O.~Edfors, and F.~Tufvesson, ``Massive {MIMO} extensions
  to the {COST} 2100 channel model: Modeling and validation,'' \emph{{IEEE}
  Trans. Wireless Commun.}, vol.~19, no.~1, pp. 380--394, Jan. 2020.

\bibitem{ZS13}
J.~Ziniel and P.~Schniter, ``Efficient high-dimensional inference in the
  multiple measurement vector problem,'' \emph{{IEEE} Trans. Signal Process.},
  vol.~61, no.~2, pp. 340--354, Jan. 2013.

\bibitem{XWG20}
X.~Xie, Y.~Wu, J.~Gao, and W.~Zhang, ``Massive unsourced rrandom access for
  massive {MIMO} correlated channels,'' in \emph{Proc. {IEEE} Glob. Commun.
  Conf. (GLOBECOM)}, Dec. 2020, pp. 1--6.

\bibitem{COV14}
L.~Chiantini, G.~Ottaviani, and N.~Vannieuwenhoven, ``An algorithm for generic
  and low-rank specific identifiability of complex tensors,'' \emph{SIAM
  Journal on Matrix Analysis and Applications}, vol.~35, no.~4, pp. 1265--1287,
  Mar. 2014.

\bibitem{G11}
W.~Gautschi, ``The incomplete {Gamma} functions since {Tricomi},'' \emph{Atti
  dei Convegni Linci}, no. 1998, pp. 203--237, 2011.

\end{thebibliography}

\ifCLASSOPTIONcaptionsoff
\newpage
\fi

\end{document}